\documentclass[times,11pt, onecolumn, draftclsnofoot]{IEEEtran}
\usepackage{color}
\usepackage[dvips]{epsfig}
\usepackage{relsize}
\usepackage{amsmath, amsthm, amssymb}
\usepackage{multirow}
\usepackage{dsfont}
\usepackage{algorithmic}
\usepackage{algorithm}
\usepackage{multirow}
\usepackage{url}

\newtheorem{theorem}{Theorem}[section]
\newtheorem{lemma}[theorem]{Lemma}

\newcommand{\snr}{\text{\relscale{.8}\textsf{SNR}}}

\begin{document}
%\title{Cooperative Spectrum Sensing in Cognitive Radio Networks: Cost Optimization}
\title{\huge{Spectrum Sensing in Cognitive Radio Networks: Performance Evaluation and Optimization}\footnote{This work was done when the first author was at Lehigh University. This work was supported by NSF Grant CNS-0721445 and CNS-0721433.}}
\author{\IEEEauthorblockN{\normalsize {Gang Xiong}\IEEEauthorrefmark{1}, \normalsize{Shalinee Kishore}\IEEEauthorrefmark{2} \normalsize{and Aylin Yener}\IEEEauthorrefmark{3}\\
\IEEEauthorblockA{\IEEEauthorrefmark{1}\normalsize{Intel Corporation, Hillsboro, OR 97124} \\ \normalsize{Email: gang.xiong@intel.com}} \\
\IEEEauthorblockA{\IEEEauthorrefmark{2}\normalsize{Electrical and Computer Engineering, Lehigh University, Bethlehem, PA 18015} \\ \normalsize{Email: skishore@lehigh.edu}} \\
\IEEEauthorblockA{\IEEEauthorrefmark{3}\normalsize{Electrical Engineering, Pennsylvania State University, University Park, PA 16802} \\ \normalsize{Email: yener@ee.psu.edu}}\vspace{-0.20in}}}
\maketitle
%\vspace{-0.1in}
\begin{abstract}
This paper studies cooperative spectrum sensing in cognitive radio networks where secondary users collect local energy statistics and report their findings to a secondary base station, i.e., a fusion center. First, the average error probability is quantitively analyzed to capture the dynamic nature of both observation and fusion channels, assuming fixed amplifier gains for relaying local statistics to the fusion center. Second, the system level overhead of cooperative spectrum sensing is addressed by considering both the local processing cost and the transmission cost. Local processing cost incorporates the overhead of sample collection and energy calculation that must be conducted by each secondary user; the transmission cost accounts for the overhead of forwarding the energy statistic computed at each secondary user to the fusion center. Results show that when jointly designing the number of collected energy samples and transmission amplifier gains, only {\em one} secondary user needs to be actively engaged in spectrum sensing. Furthermore, when number of energy samples or amplifier gains are fixed, closed form expressions for optimal solutions are derived and a generalized water-filling algorithm is provided.
\end{abstract}
\vspace{-0.15in}
%-----------------------------------------
\section{Introduction}
%-----------------------------------------

To alleviate inefficient allocation of radio frequency (RF) spectrum, cognitive radios have recently been proposed to coexist with primary (or licensed) users of spectral bands while not causing harmful interference \cite{Haykin_2005}\cite{Ne_2011}.  Current proposals for secondary networks require cognitive users to conduct {\em spectrum sensing} so that they can detect unused spectral bands and avoid interfering with a primary system.  To improve detection reliability in fading conditions, multiple secondary users can cooperate in spectrum sensing and take advantage of spatial diversity \cite{Quan_2008}\cite{Ma_2009}. 

In secondary networks where users communicate with a local secondary base station as illustrated in Fig.~\ref{fig:Coop_Specturm}, the system level performance and design of cooperative spectrum sensing must 1) account for the dynamic nature of both the observation and fusion channels, i.e., the channel between the secondary and primary users and the channel between the secondary user and the secondary base station, respectively; and 2) balance the gains offered by spectrum sensing against its computational and transmission costs.  In this paper, we address both these concerns in evaluating and designing spectrum sensing schemes for secondary networks.  
\vspace{-0.1in}
\begin{figure}[htb]
 \centerline{\epsfxsize 2.5in\epsfbox{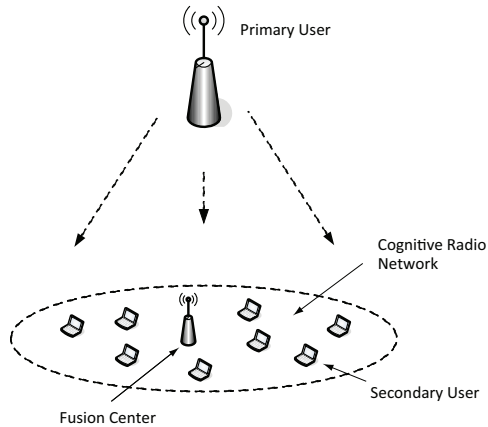}}
 \caption {Topology of cooperative spectrum sensing in cognitive radio networks.}\label{fig:Coop_Specturm} \vspace{-0.1in}
\end{figure}

%Prior studies have examined cooperative spectrum sensing under the model described in Fig.~\ref{fig:Coop_Specturm}.  
In \cite{Gha_2005}, a logic OR fusion rule for hard-decision combining was presented to cooperatively detect the primary user. The AF cooperative strategy was used in \cite{Gane_2007} to improve spectrum agility and allow two secondary users to communicate with each other. An optimal linear detector for cooperative spectrum sensing was proposed in \cite{Quan_2008}, where the received signals at the fusion center were optimally weighted for global fusion. In \cite{Unni_2008}, a linear quadratic fusion rule based on a detection criterion was proposed for spectrum sensing by modeling received signals as correlated log-normal random variables. Based on our knowledge, these and other prior studies
do not focus on system-level performance of cooperative spectrum sensing that accounts for the
dynamic nature of {\em both} the observation and fusion channels. %In contrast to these studies, we then focus on the system-level performance of cooperative spectrum sensing that account for the dynamic nature of {\em both} the observation and fusion channels. 

Low-energy overhead cooperative spectrum sensing was studied in \cite{Zhang_2009}. Optimally allocated powers were computed without taking into account the underlying system level cost of sensing. Our work on energy-constrained spectrum sensing is motivated by \cite{Appad_2005}, where detection problems accounted for constraints on expected cost due to transmission and measurement. We build on these formulations here to design energy-constrained cooperative spectrum sensing. 

In our system model, secondary users forward local energy statistics to a secondary base station using amplify-and-forward (AF) over parallel access channels.  We first address the impact of dynamic observation and fusion channels by analyzing the average error probability for cooperative spectrum sensing considering both additive white Gaussian noise (AWGN) and Rayleigh fading conditions.  Results show that detection performance can be maintained in the low and moderate fusion signal-to-noise ratio (SNR) regimes when fusion channels are reliable, whereas fading on the secondary users' observation channels provide spatial diversity. %where as fading observation channels of the secondary users can provide spatial diversity.

Next, we address the {\em system level} energy cost of sensing by considering two major factors:  Local processing cost due to sample collection and local energy calculation and transmission cost due to forwarding local statistics to the fusion center. We present two optimization problems to find the number of energy samples that must be collected at each secondary user and the appropriate amplifier gain that each secondary user must use for AF relaying of the local energy statistic.  When jointly optimizing both the number of samples and amplifier gains, we show that only {\em one} secondary user must be actively engaged in spectrum sensing.  When either the amplifier gains or the number of samples is fixed, we find closed-form optimal solutions and propose a generalized water-filling approach to energy-constrained cooperative spectrum sensing.

The remainder of the paper is organized as follows: Section II describes our system model. Section III presents the average error probability for various observation and fusion channel conditions. Sections IV and V collectively present our results for energy-constrained spectrum sensing: Section IV addresses the optimization for minimization of global error probability while Section V provides the optimization for minimization of system level cost. Simulation results are presented in Section VI and we conclude the paper in Section VII.

In this paper, we use the following notation: column vectors are denoted by boldface lowercase letters, i.e.,
$\boldsymbol{x} = [x_1, x_2, \cdots, x_{n}]^{\texttt{T}}$ and $x_i$ is the $i$th entry of $\boldsymbol{x}$. $\boldsymbol{0} = [0, 0, \cdots, 0]^\texttt{T}$ and $\boldsymbol{1} = [1, 1, \cdots, 1]^\texttt{T}$. $\mathbf{I}$ is the identity matrix. $(\cdot)^{\texttt{T}}$ and $(\cdot)^{\dagger}$ denote the transpose and conjugate transpose operation, respectively. $\|\boldsymbol{x}\|$ denotes the $\ell_2$ norm of $\boldsymbol{x}$. $\boldsymbol{x} \succeq \boldsymbol{0}$ denotes the generalized inequality, i.e., $x_i \geq 0$. $\mathcal{Z}_{+}^n$ and $\mathcal{R}_{+}^n$ denote the set of nonnegative integer and real $n$-vectors, respectively. $|\mathcal{S}|$ denotes the cardinality of a set $\mathcal{S}$. %$\mathbb{E}\{\cdot\}$ and $\mathbb{V}\textrm{ar}\{\cdot\}$ denote the expected and variance value, respectively. 
$\lceil \cdot \rceil$ and $\lfloor \cdot \rfloor$ denote the ceiling and floor operations, respectively.

%\vspace{-0.15in}
%-----------------------------------------
\section{System Model}
%-----------------------------------------
%-----------------------------------------
\subsection{Communication Model}
%-----------------------------------------

We consider a network model in Fig.~\ref{fig:Coop_Specturm}, where secondary user conducts local spectrum sensing and   
transmits its local energy statistic to the fusion center using AF on parallel access channels (PAC). The received signal for secondary user $i$ at the fusion center is shown in  Fig.~{\ref{fig:PAC_CSS_AWGN}}, i.e.,
\begin{flalign}\label{eq:fusion_model}
y_i = {g_i h_i} x_i +v_i,
\end{flalign}
where $x_i$ is the energy of received signal at the secondary user $i$; ${g}_i$ is the amplifier gain for the secondary user $i$; $h_i$ is the channel gain between secondary user $i$ and the fusion center and $v_i$ is independent and identically distributed (i.i.d.) Gaussian noise, i.e., $v_i \sim \mathcal{{CN}}(0, \sigma_v^2)$ and is independent of $x_i$. We assume that $h_i$ is known at
the fusion center (e.g., via channel estimation) and remains
constant during the sensing period. %It is interesting to note that we may utilize beamforming strategy as an alternate approach of data transmission for the secondary users. In this scheme, the secondary user forwards the local energy to the fusion center simultaneously with beamforming weights. By assigning optimal beamforming weight for each secondary user, we note that this stratergy can lead to the same optimal fusion rule as PAC channels for cooperative spectrum sensing. However, due to practical implementation issues (i.e., phase and timing synchronization), we focus out study on PAC channels in this paper. 
We can then rewrite (\ref{eq:fusion_model}) in a matrix form as
\begin{flalign}
\boldsymbol{y} = \mathbf{H} \boldsymbol{x} + \boldsymbol{v},
\end{flalign}
where $\mathbf{H} = \textrm{diag}\{{g_1 h_{1}}, {g_2 h_{2}},\cdots, {g_{n} h_{n}} \} $.

\begin{figure}[htb]
 \centerline{\epsfxsize 4.5in\epsfbox{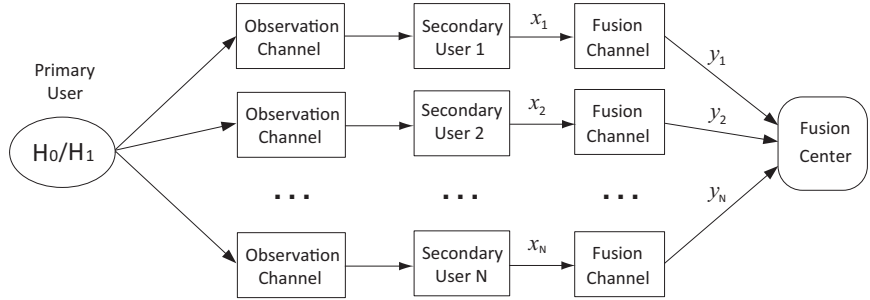}}
 \caption {Cooperative spectrum sensing in cognitive radio networks.}\label{fig:PAC_CSS_AWGN}\vspace{-0.2in}
\end{figure}

\vspace{-0.15in}

%-----------------------------------------
\subsection{Local Energy Statistic}
%-----------------------------------------

For secondary user $i,~(1 \leq i \leq n)$, the hypothesis test for $x_i$ is given as
\begin{equation}
\left\{
\begin{array}{ll}
\mathcal{H}_0: & x_i = (1/\kappa_i)\sum_{k=1}^{\kappa_i} |n_i(k)|^2  \\
\mathcal{H}_1: & x_i = (1/\kappa_i)\sum_{k=1}^{\kappa_i} |\tilde{h}_i s(k)+n_i(k)|^2,
\end{array} \right.
\end{equation}
where $\kappa_i$ is the number of samples, $s(k)$ is the transmitted signal from the primary user and $n_i(k)$ is the noise received by secondary user $i$. We assume $s(k)$ is complex PSK modulated and i.i.d. with mean zero and variance $\sigma_s^2$; $\tilde{h}_i$ is the channel gain between the primary user and secondary user $i$ and is assumed to be constant during the cooperative spectrum sensing period; and $n_i(k)$ is i.i.d. Gaussian noise with mean zero and variance $\sigma_n^2$ and is independent of $s(k)$. We define the local received SNR at the secondary user $i$ as
$
\gamma_{i} = \sigma_s^2|\tilde{h}_i|^2/\sigma_n^2.
$ 
When $\kappa_i$ is large, $x_i$ can be approximated as Gaussian random variable \cite{Quan_2008}, i.e.,
\begin{equation}\label{eq:system_model}
\left\{
\begin{array}{ll}
\mathcal{H}_0: &  x_i \sim \mathcal{{N}} ( \sigma_n^2, ~\sigma_n^4/\kappa_i)  \\
\mathcal{H}_1: &  x_i \sim \mathcal{{N}} ((1+\gamma_{i})\sigma_n^2 , ~(1+2\gamma_{i}) \sigma_n^4/\kappa_i ).
\end{array} \right.
\end{equation}
We assume here the local received SNR $\gamma_i$ is known at secondary user $i$. In IEEE 802.22, $\gamma_i$ can be estimated from pilot signals periodically transmitted by primary users \cite{Quan_2009}\footnote{Uncertainty in the knowledge of local received SNR would affect the design of cooperative spectrum sensing. We will investigate this important issue in the future.}.

Given this system model, we see that
%\begin{flalign*}
%\xi_i \overset{\underset{\mathrm{def}}{}}{=}~&\mathbb{E}\{x_i^2\} \\ 
%=~&  \pi_0 (1+1/\kappa_i)\sigma_n^4 + \pi_1 \left[(1+\gamma_i)^2 + (1+2\gamma_i)/\kappa_i\right] \sigma_n^4 \\
%=~& \left[ 1+ 1/\kappa_i + \pi_1 \left(\gamma_i + 2\left(1+1/\kappa_i\right) \right)\gamma_i\right]  \sigma_n^4,
%\end{flalign*}
$
\xi_i \overset{\underset{\mathrm{def}}{}}{=}\mathbb{E}\{x_i^2\} = \left[ 1+ 1/\kappa_i + \pi_1 \left(\gamma_i + 2\left(1+1/\kappa_i\right) \right)\gamma_i\right]  \sigma_n^4,
$
where $\pi_0 = \textrm{P}(\mathcal{H}_0)$ and $\pi_1 = \textrm{P}(\mathcal{H}_1)$ are the probabilities that spectrum is idle and occupied, respectively.
In cognitive radio networks, the received primary user power measured by the secondary user is expected to be very small \cite{Sah_2004}, i.e., $\gamma_i \ll 1$. Additionally, the number of samples is expected to be more than a few, i.e., $\kappa_i \gg 1$. Thus, we can approximate the transmitted power for the secondary user $i$ as
$
\mathcal{P}_i = \xi_i g_i^2  \simeq g_i^2 (1+2\pi_1 \gamma_i) \sigma_n^4.
$

%-----------------------------------------
\subsection{Optimal Fusion Rule}
%-----------------------------------------

Under hypothesis $\mathcal{H}_0$ and $\mathcal{H}_1$, the received signal $\boldsymbol{y} $ has a Gaussian distribution, i.e.,
\begin{equation}\label{eq:y_dist}
\left\{
\begin{array}{ll}
\mathcal{H}_0: &  \boldsymbol{y} \sim \mathcal{{N}} \left( \mathbf{H}\boldsymbol{1} \sigma_n^2, ~ \mathbf{\Sigma}_0\right)  \\
\mathcal{H}_1: &  \boldsymbol{y} \sim \mathcal{{N}} \left(\mathbf{H}(\boldsymbol{1}+\boldsymbol{\gamma})\sigma_n^2, ~  \mathbf{\Sigma}_1\right),
\end{array} \right.
\end{equation}
where $\mathbf{\Sigma}_0 = \mathbf{H} \mathbf{S}\mathbf{H}^{\dagger} \sigma_n^4 +  \sigma_v^2 \mathbf{I}$ and $\mathbf{\Sigma}_1 = \mathbf{H} \mathbf{S}(\mathbf{I}+2\mathbf{\Gamma})\mathbf{H}^{\dagger}  \sigma_n^4 +  \sigma_v^2 \mathbf{I}$, here, $\mathbf{\Gamma} = \textrm{diag}\{\gamma_1, \gamma_2, \cdots, \gamma_{n} \}$ and $\mathbf{S} = \textrm{diag}\{1/\kappa_1, 1/\kappa_2, \cdots, 1/\kappa_{n} \}$.  
Without loss of generality, assume that $\pi_0 = \pi_1 = 0.5$. Then, optimal (maximum a posteriori probability) likelihood ratio test (LRT) is given as:%\footnote{We note that we can reach same optimization formulation by using Neyman-Pearson criterion to maximize global detection probability. Here we present the global error probability for the sake of simplicity.}:
 \begin{flalign}
\log \frac{p(\boldsymbol{y}|\mathcal{H}_1)}{p(\boldsymbol{y}|\mathcal{H}_0)} ~\mathop{\gtrless}_{\mathcal{H}_0}^{\mathcal{H}_1}~ 0.
\end{flalign}

Since $\gamma_i \ll 1$ and  $\kappa_i \gg 1$, then, $\gamma_i/\kappa_i \approx 0$ and we have $\mathbf{\Sigma}_0  \approx \mathbf{\Sigma}_1$. Thus, the optimal LRT can be approximated as
 \begin{flalign}
 \mathcal{T}(\boldsymbol{y}) ~=~ (\mathbf{H}\boldsymbol{\gamma})^{\dagger} \mathbf{\Sigma}_0^{-1}\boldsymbol{y}~\mathop{\gtrless}_{\mathcal{H}_0}^{\mathcal{H}_1}~ \tau,
\end{flalign}
where $\tau \!=\! (\mathbf{H}\boldsymbol{\gamma})^{\dagger} \mathbf{\Sigma}_0^{-1}\mathbf{H}(\boldsymbol{1}+0.5\boldsymbol{\gamma})\sigma_n^2$. Furthermore, we note that
$\mathbb{E}\{\mathcal{T}(\boldsymbol{y})|\mathcal{H}_0\} ~=~ (\mathbf{H}\boldsymbol{\gamma})^{\dagger} \mathbf{\Sigma}_0^{-1}\mathbf{H}\boldsymbol{1} \sigma_n^2$, $\mathbb{E}\{\mathcal{T}(\boldsymbol{y})|\mathcal{H}_1\} ~=~ (\mathbf{H}\boldsymbol{\gamma})^{\dagger} \mathbf{\Sigma}_0^{-1}\mathbf{H}(\boldsymbol{1}+\boldsymbol{\gamma})\sigma_n^2$ and $\mathbb{V}\textrm{ar} \{\mathcal{T}(\boldsymbol{y})|\mathcal{H}_0\}  ~=~ \mathbb{V}\textrm{ar} \{\mathcal{T}(\boldsymbol{y})|\mathcal{H}_1\}  ~=~ (\mathbf{H}\boldsymbol{\gamma})^{\dagger} \mathbf{\Sigma}_0^{-1}\mathbf{H}\boldsymbol{\gamma}$.
%\begin{equation*}
%\left\{
%\begin{array}{l}
%\mathbb{E}\{\mathcal{T}(\boldsymbol{y})|\mathcal{H}_0\} ~=~ (\mathbf{H}\boldsymbol{\gamma})^{\dagger} \mathbf{\Sigma}_0^{-1}\mathbf{H}\boldsymbol{1} \sigma_n^2 \\
%\mathbb{E}\{\mathcal{T}(\boldsymbol{y})|\mathcal{H}_1\} ~=~ (\mathbf{H}\boldsymbol{\gamma})^{\dagger} \mathbf{\Sigma}_0^{-1}\mathbf{H}(\boldsymbol{1}+\boldsymbol{\gamma})\sigma_n^2,
%\end{array} \right.
%\end{equation*}
%and 
%\[
%\mathbb{V}\textrm{ar} \{\mathcal{T}(\boldsymbol{y})|\mathcal{H}_0\}  ~=~ \mathbb{V}\textrm{ar} \{\mathcal{T}(\boldsymbol{y})|\mathcal{H}_1\}  ~=~ (\mathbf{H}\boldsymbol{\gamma})^{\dagger} \mathbf{\Sigma}_0^{-1}\mathbf{H}\boldsymbol{\gamma}.
%\]

With this preparation, it can be shown that the error probability is given as\footnote{It is worth mentioning that we can reach same optimization formulation by using Neyman-Pearson criterion to maximize global detection probability. Here we present the global error probability for the sake of simplicity.}
%\begin{flalign}\label{eq:error_prob_given_h}
% \textrm{P}_{e} ~=~& \pi_0  \textrm{P}_{f} + \pi_1  \textrm{P}_{m} \nonumber \\
% ~=~& Q \left( \frac{\sigma_n^2}{2}\left[ (H\boldsymbol{\gamma})^{\dagger} \mathbf{\Sigma}_0^{-1}H\boldsymbol{\gamma} \right]^{1/2} \right) \nonumber \\
% ~=~& Q \left( \frac{1}{2}\sqrt{\mathcal{F}( \boldsymbol{\kappa}, \boldsymbol{g}) } \right),
%\end{flalign}
\begin{flalign}\label{eq:error_prob_given_h}
 \textrm{P}_{e} ~=~ \pi_0  \textrm{P}_{f} + \pi_1  \textrm{P}_{m}  ~=~ Q \left( \frac{1}{2}\bigg(\sum_{i=1}^{n} \frac{   g_i^2 \kappa_i \gamma_i^2 |h_{i}|^2 }{g_i^2 |h_{i}|^2 + \kappa_i\tilde{\sigma}_v^2} \bigg)^{1/2}\right),
\end{flalign}
where $\tilde{\sigma}_v^2 = \sigma_v^2/\sigma_n^4$ and $
Q(x) = \frac{1}{\sqrt{2 \pi}} \int_x^{\infty} \exp (-t^2/2) \textrm{d}t.
$
%and 
%\[
%\mathcal{F}( \boldsymbol{\kappa}, \boldsymbol{g}) =  \sum_{i=1}^{n} \frac{   g_i^2 \kappa_i \gamma_i^2 |h_{i}|^2 }{g_i^2 |h_{i}|^2 + \kappa_i\tilde{\sigma}_v^2},
%\]
%where . 
%From (\ref{eq:error_prob_given_h}), we note that the error probability is a decreasing function of the number of samples or amplifier gain for secondary user $i$, i.e., when $\kappa_i$ or $g_i$ increases, the error probability $\textrm{P}_{e}$ decreases.  
It is also easy to see that the asymptotic error probability expressions when the number of samples or amplifier gains approach infinity are given by
\begin{flalign}\label{eq:asym_kappa}
 \textrm{P}_{e}({\kappa}_{\infty}) ~\overset{\underset{\mathrm{def}}{}}{=}~  \underset{{\kappa_i \rightarrow \infty}}{\lim} \textrm{P}_{e}  ~=~  Q \left( \frac{1}{2 \tilde{\sigma}_v}\bigg( \sum_{i=1}^{n}  {g}_i^2 \gamma_i^2 |h_{i}|^2 \bigg)^{1/2} \right),
\end{flalign}
and
\begin{flalign}\label{eq:asym_g}
 \textrm{P}_{e}({g}_{\infty})~ \overset{\underset{\mathrm{def}}{}}{=} ~\underset{{g_i \rightarrow \infty}}{\lim} \textrm{P}_{e} ~=~ Q \left( \frac{1}{2 }\bigg({ \sum_{i=1}^{n}\kappa_i  \gamma_i^2 \bigg)^{1/2}} \right),
\end{flalign}
respectively.% It is interesting to note that when $\kappa_i \rightarrow \infty$, the asymptotic error probability is a function of local received SNR and fusion channel gains; while when $g_i \rightarrow \infty$, the asymptotic error probability is only a function of local received SNR.

%-----------------------------------------
\subsection{System Level Cost for Cooperative Spectrum Sensing}
%-----------------------------------------

In this paper, we consider system level cost for cooperative spectrum sensing in cognitive radio networks. This system level cost has contributions from three components: Local processing; transmission; and reporting and broadcasting.
\begin{itemize}
	\item Local processing cost includes the energy consumed by the secondary user in receiver RF scanning and local energy calculation. For simplicity, we assume that the local processing cost $\mathcal{C}_{pi}(\cdot)$ for secondary user $i$ is a linear function of the number of samples \cite{Hoang_2009}, i.e.,
$
\mathcal{C}_{pi}(\kappa_i) = c_0 \kappa_i,
$
where $c_0$ is the local processing cost per sample.
\item Transmission cost is the transmit power required from a secondary user to
transmit the local calculated energy to the fusion center.
Here, we assume that this cost for secondary user $i$ is given as
%
%
%
%Transmission cost is the cost for the secondary user to transmit the local calculated energy to the fusion centers. Here, we assume that the transmission cost $\mathcal{C}_{ti}(\cdot)$ for secondary user $i$ is equal to the transmitted power, i.e.,
$
\mathcal{C}_{ti}(g_i) = \mathcal{P}_i = \xi_i g_i^2.
$
\item %For optimal system design, fusion center needs to know the local received SNR for each secondary user. In practice, this can be realized for secondary users to report their local received SNRs to the fusion center. The fusion center then determines the resource allocated to each secondary user, and broadcasts it to the secondary users. In this paper, we assume that total reporting and broadcasting cost $\mathcal{C}_{rb}$ is fixed, and thus can be eliminated from the system level cost.
For optimal system design, the fusion center needs to know
the local received SNR for each secondary user. In
practice, this means that secondary users will
report their local received SNRs to the fusion center. The
fusion center then determines optimal allocations (number of samples and/or amplifier gains) to each secondary user and then broadcasts them to all secondary
users. In this paper, we assume that this total reporting and
broadcasting cost $\mathcal{C}_{rb}$ is fixed; thus we do not consider it in the optimization problem.
\end{itemize}

The system level cost during the cooperative spectrum sensing (aside from $\mathcal{C}_{rb}$) is given as
\begin{flalign*}
\mathcal{C}(\boldsymbol{\kappa}, \boldsymbol{g}) ~=~  \sum_{i=1}^{n} \mathcal{C}_{pi}(\kappa_i) +   \sum_{i=1}^{n} \mathcal{C}_{ti}(g_i) ~=~  \sum_{i=1}^{n} \left(c_0\kappa_i +   \xi_i g_i^2 \right).
\end{flalign*}

%-----------------------------------------
\section{Average Error Probability }\label{sec:avg_error_prob}
%-----------------------------------------

%In this section, we provide the average error probability for spectrum sensing.  We assume 
We assume in this section that the amplifier gains and the number of samples collected at each secondary user are fixed and not adjusted according to the channel gains. We will discuss adapting amplifier gains and number of samples in the subsequent sections.
%%-----------------------------------------
%\subsection{Exact Average Error Probability}%\label{sec:avg_error_prob}
%%-----------------------------------------
From (\ref{eq:error_prob_given_h}), we see that the average error probability can be calculated as
 \begin{flalign}\label{eq:avg_error_prob}
 \textrm{P}_{e, \textrm{avg}}  ~=~  \mathbb{{E}}_{\boldsymbol{\gamma}, \boldsymbol{h}} \left\{\textrm{P}_{e|\boldsymbol{\gamma}, \boldsymbol{h}} \right\}.
\end{flalign}

To simplify the calculation
of the average error probability, we consider the following alternate
expression for the Q function \cite{Alouini_1999}
$
Q(x) = \frac{1}{\pi} \int_{0}^{\pi/2} \exp\Big(-\frac{x^2}{ 2\sin^2 \phi} \Big)\textrm{d}\phi,~x \geq 0.
$
When the local received SNRs $\gamma_i$ and fusion channel gains $|h_{i}|^2$ are independent, respectively, we can simplify the average error probability in (\ref{eq:avg_error_prob}) as 
 \begin{flalign} 
 \textrm{P}_{e, \textrm{avg}}  ~=~  \frac{1}{\pi} \int_{0}^{\pi/2} \prod_{i=1}^{n} \mathcal{B}_i(\phi)\textrm{d}\phi,
\end{flalign}
where 
$
\mathcal{B}_i(\phi) = \int_{0}^{\infty}\int_{0}^{\infty}\exp\left(\frac{\mathcal{A}_i(s, t)}{\sin^2 \phi} \right) p_{\gamma_i}(s)p_{|h_{i}|^2}(t)\textrm{d}s \textrm{d} t
$
and
$
\mathcal{A}_i(s, t) = - \frac{ 1}{ 8}\cdot \frac{g_i^2 \kappa_i s^2 t}{g_i^2 t + \kappa_i \tilde{\sigma}_v^2}.
$
Here, $p_{\gamma_i}(s) $ and $p_{|h_{i}|^2}(t)$ are PDFs of $\gamma_i$ and $|h_{i}|^2$, respectively. If we further assume that $g_i = g$, $\kappa_i = \kappa$, $\gamma_i$ and $|h_{i}|^2$ are i.i.d., respectively, i.e., $p_{\gamma_i}(s) = p_{\gamma}(s)$ and $p_{|h_{i}|^2}(t) = p_{|h|^2}(t)$, then we have $\mathcal{B}_i(\phi) = \mathcal{B}(\phi),~\forall i$. In this case, the average error probability in (\ref{eq:avg_error_prob}) reduces to 
 \begin{flalign}\label{eq:avg_simplify}
 \textrm{P}_{e, \textrm{avg}}  ~=~  \frac{1}{\pi} \int_{0}^{\pi/2} \left[ \mathcal{B} (\phi) \right]^{n} \textrm{d}\phi.
\end{flalign}

Based on this, we see that ${ \textrm{P}}_{e, \textrm{avg}} $ is a decreasing function of $n$, which indicates that in a power unconstrained cognitive radio network, global error performance can be improved by increasing the number of secondary users. This statement follows since $\mathcal{A}(s, t) \leq 0$ and $ \mathcal{B} (\phi) \leq \int_{0}^{\infty}\int_{0}^{\infty} p_{\gamma}(s)p_{|h|^2}(t)\textrm{d}s \textrm{d} t = 1$.
In general, a closed-form expression of $ \textrm{P}_{e, \textrm{avg}}$ is difficult to obtain. However, only elementary functions, such as exponential and $Q(\cdot)$, are involved in the integral calculation; the average error probability can thus readily be found numerically. 

%%-----------------------------------------
%\subsection{Upper Bound for Average Error Probability}%\label{sec:avg_error_prob}
%%-----------------------------------------
 \emph{Remark}: To gain more insight, we investigate an upper bound for
average error probability. Since $Q(x) \leq \frac{1}{2}\exp (-x^2/2)$, the upper bound can be obtained as
 \begin{flalign} 
 \tilde{ \textrm{P}}_{e, \textrm{avg}}  ~=~   \frac{1}{2}   \prod_{i=1}^{n} \mathcal{M}_i,
\end{flalign}
where 
$
 \mathcal{M}_i = \int_{0}^{\infty}\int_{0}^{\infty}\exp[{\mathcal{A}_i(s, t)} ] p_{\gamma_i}(s)p_{|h_{i}|^2}(t)\textrm{d}s \textrm{d} t.
$
Assume for simplicity $g_i = g,~\kappa_i = \kappa $, $\gamma_i$ and $|h_{i}|^2$ are i.i.d., respectively. Then, we have $ \mathcal{M}_i =  \mathcal{M},~\forall i $, and 
 \begin{flalign}\label{eq:up_simplify}
\tilde{ \textrm{P}}_{e, \textrm{avg}}  ~= ~ \frac{1}{2}  \mathcal{M}^{n}.
\end{flalign}
It is readily evident that when $g \rightarrow 0$, $\tilde{ \textrm{P}}_{e, \textrm{avg}} \rightarrow \frac{1}{2}$. This is not surprising since when the amplifier gains are low, the fusion center will not be able to make a global decision due to the lack of local energy statistic. %In the later case, we denote the asymptotic upper bounds for three scenarios as $\tilde{ \textrm{P}}_{e, \textrm{avg}}^{(1)}(g_\infty)$, $\tilde{ \textrm{P}}_{e, \textrm{avg}}^{(2)}(g_\infty)$ and $\tilde{ \textrm{P}}_{e, \textrm{avg}}^{(3)}(g_\infty)$, respectively.
%%-----------------------------------------
%\subsection{Closed-Form Expressions for Three Scenarios}%\label{sec:avg_error_prob}
%%-----------------------------------------
Next, we use (\ref{eq:avg_simplify}) and (\ref{eq:up_simplify}) to evaluate the average error probability for cooperative spectrum sensing for the three channel scenarios shown in Table \ref{tb1:3_scenarios}.

\begin{table} [h] \caption{Three Channel Environments for Performance Evaluation} \label{tb1:3_scenarios}
\vspace{-0.15in}
\begin{center}
\begin{tabular}{c|c|c}  \hline
& Observation channels  & Fusion channels   \\   \hline
Channel Environment I & AWGN   & Rayleigh fading   \\  \hline
Channel Environment II & Rayleigh fading   & AWGN  \\  \hline
Channel Environment III & Rayleigh fading  & Rayleigh fading  \\  \hline
\end{tabular} 
\end{center} 
\end{table} 
\vspace{-0.15in}

%%-----------------------------------------
\subsubsection{Channel Environment I}
%%----------------------------------------- 

In this scenario, $\gamma_i = \bar{\gamma}$ and $p_{|h_{i}|^2}(t) =  \exp(- t)$ since the observation channel is AWGN and the fusion channel is exponential Rayleigh fading. After some manipulations, we have
\[
\mathcal{B}(\phi) = \exp\left(-\frac{\kappa \bar{\gamma}^2}{8 \sin^2 \phi} \right) {\Psi}_1\left(\frac{\kappa \bar{\gamma}^2}{8 \sin^2 \phi}, \frac{\kappa\tilde{\sigma}_v^2}{g^2} \right),
\]
where 
$
\Psi_1\left(a,b\right) = \int_{0}^{\infty} \exp\big(-x + \frac{ab}{x+b} \big)\textrm{d}x,~(a,b>0).
$
After calculating $\mathcal{B}(\phi)$, we substitute it in (\ref{eq:avg_simplify}) to obtain the average error probability.
It is interesting to note that a similar definition of $\Psi_1\left(\phi,a,b\right)$ can be found in \cite{Gane_2007}. 
Furthermore, the upper bound is given as
\[
\tilde{ \textrm{P}}_{e, \textrm{avg}}^{(1)} =   \frac{1}{2} \exp\left( -\frac{n\kappa \bar{\gamma}^2 }{8}\right) \left[ {\Psi}_1\Big(\frac{\kappa \bar{\gamma}^2}{8 }, \frac{\kappa\tilde{\sigma}_v^2}{g^2} \Big) \right]^n.
\]
When $g\rightarrow \infty$, we see that $\tilde{ \textrm{P}}_{e, \textrm{avg}}^{(1)}(g_\infty) = \frac{1}{2} \exp\left( -\frac{n\kappa \bar{\gamma}^2 }{8}\right)$. This indicates that when the fusion channel is perfect, average error performance is limited by local observed energy statistic. 

%%-----------------------------------------
\subsubsection{Channel Environment II}\label{sec:avg_caseII}
%%----------------------------------------- 

In this scenario, $p_{\gamma_i}(s) =  \frac{1}{\bar{\gamma}} \exp(- \frac{s}{\bar{\gamma}})$ and $h_{i} = 1$.
After some manipulations, (using eq.(3.322.2) in \cite{Integrals}), we obtain
\[
\mathcal{B}(\phi) = \sqrt{8 \pi c}\sin \phi  \exp\left(2 c \sin^2 \phi\right) Q\left(2\sqrt{c} \sin \phi\right),
\]
where 
$
c \! =\! \frac{1}{\bar{\gamma}^2} \Big(\frac{1}{\kappa} + \frac{\tilde{\sigma}_v^2}{g^2 }\Big).
$
Furthermore, the upper bound is
\[
\tilde{ \textrm{P}}_{e, \textrm{avg}}^{(2)}\! =\!   \frac{1}{2}\big({8 \pi c}\big)^{n/2}  \exp\big(2 n c \big) \big[Q\left(2\sqrt{c}\,  \right)\big]^{n}.
\]
When $g\rightarrow \infty$, we see that $c \rightarrow 1/(\kappa\bar{\gamma}^2)$ and 
$
\textstyle \tilde{ \textrm{P}}_{e, \textrm{avg}}^{(2)}(g_\infty) = \frac{1}{2}\left(\frac{8 \pi}{\kappa \bar{\gamma}^2}\right)^{n/2}  \exp\left( \frac{2n }{\kappa \bar{\gamma}^2} \right) \left[Q\left(\frac{2}{\bar{\gamma} \sqrt{\kappa}}  \right)\right]^{n}.
$
Again, we see that the average error performance is limited by local observed energy statistic when $g\rightarrow \infty$.

%%-----------------------------------------
\subsubsection{Channel Environment III}
%%----------------------------------------- 

In this scenario, $p_{\gamma_i}(s) =  \frac{1}{\bar{\gamma}} \exp(- \frac{s}{\bar{\gamma}})$ and $p_{|h_{i}|^2}(t) =  \exp(- t)$. After some manipulations, we have
\[
\mathcal{B}(\phi) = \sqrt{8\pi} \exp\left(\frac{2\sin^2 \phi}{\kappa \bar{\gamma}^2 } \right) {\Psi}_2\left( \frac{\sin^2 \phi}{\kappa \bar{\gamma}^2}, \frac{\tilde{\sigma}_v^2 \sin^2 \phi}{g^2  \bar{\gamma}^2}\right),
\]
where 
$
 \Psi_2\left( a,b\right) = \int_{0}^{\infty}\left({a + \frac{b}{x} }\right)^{1/2}\exp\left(-x +  \frac{2b }{x}\right)   Q\left(2 \left(a+ \frac{b}{x} \right)^{1/2}\right) \textrm{d}x,~(a,b>0).
%\Psi_2\left( a,b\right) = \int_{0}^{\infty}\Big({a + \frac{b}{x} }\Big)^{1/2}\exp\Big(-x +  \frac{2b }{x}\Big)   Q\Big[2 \big(a+ \frac{b}{x} \big)^{1/2}\Big] \textrm{d}x,~(a,b>0).
$
Furthermore, the upper bound is
\[
\tilde{ \textrm{P}}_{e, \textrm{avg}}^{(3)} =    \frac{1}{2} \big({8 \pi}\big)^{n/2} \exp\left( \frac{2 n }{\kappa \bar{\gamma}^2} \right)\left[ {\Psi}_2\Big(\frac{1}{\kappa \bar{\gamma}^2}, \frac{\tilde{\sigma}_v^2}{g^2 \bar{\gamma}^2} \Big) \right]^n.
\]
When $g\rightarrow \infty$, we see that 
$
\textstyle \tilde{ \textrm{P}}_{e, \textrm{avg}}^{(3)}(g_\infty) = \textstyle \tilde{ \textrm{P}}_{e, \textrm{avg}}^{(2)}(g_\infty).
$ 
This is primarily due to the fact that when $g\rightarrow \infty$, the fusion channel no longer impacts the average error performance. 

%-----------------------------------------
\section{Optimization: Minimization of Error Probability}\label{sec:primal}
%-----------------------------------------

In this section, we aim to minimize the error probability for the system model in Fig.~\ref{fig:PAC_CSS_AWGN} subject
to a system level cost constraint of sensing.  Specifically, we determine the appropriate number
of samples and amplifier gains for each secondary user and consider the following two scenarios for this
optimization problem:

%In this section, we aim to minimize the error probability for the system model in Fig.~\ref{fig:PAC_CSS_AWGN}
%subject to a constraint on the cost of sensing. Specifically, we consider the following two scenarios 
%for this optimization problem:

\begin{enumerate}
	\item \textbf{Scenario A:} First, we consider the system level cost constraint. Hence, the optimization problem is formulated as:
 \begin{flalign}\label{eq:pro_formu}
\underset{\boldsymbol{\kappa}, \boldsymbol{g}}{\min}  ~~~& \textrm{P}_{e}(\boldsymbol{\kappa}, \boldsymbol{g})
\nonumber  \\  \textrm{s.t.}  ~~~& \mathcal{C}(\boldsymbol{\kappa}, \boldsymbol{g}) \leq \bar{\mathcal{C}}, ~\boldsymbol{\kappa}\in \mathcal{Z}_{+}^n, ~\boldsymbol{g}\in \mathcal{R}_{+}^n,
\end{flalign}
where $\bar{\mathcal{C}} $ is the system level cost constraint. Here we denote the optimal solution of (\ref{eq:pro_formu}) as $\big(\kappa_{p,i}^{(\textrm{opt}, 1)}, g_{p,i}^{(\textrm{opt}, 1)}\big)$ and the minimum error probability as $\textrm{P}_{e}^{(\textrm{opt},1)}$.
	\item \textbf{Scenario B:} In some applications, local sample collection for each secondary user may be scheduled in a fixed time slot. This indicates the number of samples is upper bounded by a maximum value ${\kappa}_{\max}$. Furthermore, the transmission power for each secondary user may be required to be below a predefined power limit $\mathcal{P}_{\max}$. By incorporating these additional individual constraints imposed on each secondary user, we can model the optimization problem as 
 \begin{flalign}\label{eq:pro_formu_ind}
\underset{\boldsymbol{\kappa}, \boldsymbol{g}}{\min}  ~~~& \textrm{P}_{e}(\boldsymbol{\kappa}, \boldsymbol{g})
\nonumber  \\  \textrm{s.t.}  ~~~& \mathcal{C}(\boldsymbol{\kappa}, \boldsymbol{g}) \leq \bar{\mathcal{C}},~\boldsymbol{\kappa}\in \mathcal{Z}_{+}^n, ~\boldsymbol{g}\in \mathcal{R}_{+}^n,  \nonumber
  \\   ~~~&   \boldsymbol{\kappa}\preceq {\kappa}_{\max}\boldsymbol{1},~\xi_i g_i^2 \leq \mathcal{P}_{\max}.
\end{flalign}
 \end{enumerate}

To better understand the optimal resource allocation for cooperative spectrum sensing, we consider the following two cases in Scenarios A and B as illustrated in   Table \ref{tb2:4_scenarios}: joint optimization of $\boldsymbol{\kappa}$ and $\boldsymbol{g}$; and optimization of either $\boldsymbol{\kappa}$ or $\boldsymbol{g}$. 
\begin{table} [h] \caption{Optimization Problems for Cooperative Spectrum Sensing} \label{tb2:4_scenarios}
\vspace{-0.15in}
\begin{center}
\begin{tabular}{c|p{5cm}|p{5cm}}  \hline
& ~~~~~~~~~~~~~~~~\textbf{Case I}  & ~~~~~~~~~~~~~~~~\textbf{Case II}  \\   \hline
\textbf{Scenario A} & joint optimization of $\boldsymbol{\kappa}$ and $\boldsymbol{g}$ with system level cost constraint  &  optimization of either $\boldsymbol{\kappa}$ or $\boldsymbol{g}$ with system level cost constraint  \\  \hline
\textbf{Scenario B} & joint optimization of $\boldsymbol{\kappa}$ and $\boldsymbol{g}$ with system level and individual constraints  &  optimization of either $\boldsymbol{\kappa}$ or $\boldsymbol{g}$ with system level and individual constraints  \\  \hline
\end{tabular} 
\end{center} 
\end{table} 
\vspace{-0.25in}

%-----------------------------------------
\subsection{Case I: Joint Optimization of $\boldsymbol{\kappa}$ and $\boldsymbol{g}$}\label{sec:pri_case_III}
%-----------------------------------------

%-----------------------------------------
\subsubsection{Scenario A}\label{sec:caseIII_scenarioI}
%-----------------------------------------

In this case, we consider the optimization in (\ref{eq:pro_formu}) over both $\boldsymbol{\kappa}$ and $\boldsymbol{g}$. We note that (\ref{eq:pro_formu}) is a mixed integer nonlinear optimization problem (MINLP). In general, there is no polynomial-time algorithm for solving general MINLPs \cite{IP_Book}. A potentially clearer insight into the solutions can be obtained by considering a convex relaxation for this optimization problem, where we simply relaxed the integer constraint of the number of samples:
 \begin{flalign}\label{eq:pro_formu_relax}
\underset{\boldsymbol{\kappa}, \boldsymbol{g}}{\min}  ~~~& \textrm{P}_{e}(\boldsymbol{\kappa}, \boldsymbol{g})
\nonumber  \\  \textrm{s.t.}  ~~~& \mathcal{C}(\boldsymbol{\kappa}, \boldsymbol{g}) \leq \bar{\mathcal{C}}, ~\boldsymbol{\kappa}\in \mathcal{R}_{+}^n, ~\boldsymbol{g}\in \mathcal{R}_{+}^n.
\end{flalign}
%We will show in our simulations that the performance gap between (\ref{eq:pro_formu}) and (\ref{eq:pro_formu_relax}) is very small. Hence, we will focus on (\ref{eq:pro_formu_relax}) as our optimization problem in the following analysis.
As shown in the Appendix A, (\ref{eq:pro_formu_relax}) is a convex problem. 
%\end{lemma} 
Thus, it can be solved efficiently using interior-point methods or other iterative methods \cite{Convex}. This will be a recurring theme in the optimization problems we consider in the sequel. In the numerical results, we shall see that the approximation as detailed below results in near optimal performance without the curse of complexity. Given this convex optimization problem, first we introduce the following lemma.
\begin{lemma}\label{lemma:u_v_0_same}
Optimal solution of $(\boldsymbol{\kappa}, \boldsymbol{g})$ in (\ref{eq:pro_formu_relax}) should satisfy either 1) $\kappa_i > 0 $ and $g_i > 0 $, or 2) $\kappa_i = 0 $ and $g_i = 0 $ for secondary user $i$. 
\end{lemma} 
\begin{proof}
Please see the Appendix B.
\end{proof}

This lemma is not surprising because when one secondary user does not collect the energy samples, it will not have anything to transmit to the fusion center. Similarly, when one secondary user decides not to transmit the data to the fusion center, it is reasonable to expect that this secondary user should remain inactive and not collect local energy samples. 
Using Lemma \ref{lemma:u_v_0_same}, the optimal solution of $(\boldsymbol{\kappa}, \boldsymbol{g})$ can be found as stated in the following theorem.

\begin{theorem}\label{thm:caseIII_sol}
Consider the optimization problem in (\ref{eq:pro_formu_relax}), let us define 
$
 \rho_i = \frac{\gamma_i^2 |h_{i}|^2}{(\tilde{\sigma}_v \sqrt{\xi_i} + |h_{i}|\sqrt{c_0}\,  )^2}
$ 
and assume $\rho_1 \geq \rho_2 \geq \cdots \geq \rho_{n}$. Then, the optimal solution of $(\boldsymbol{\kappa}, \boldsymbol{g})$ is
\begin{flalign}\label{eq:opt_k_z_case_III}
\kappa_{p,i}^{(\emph{opt},2)} ~=~& \left\{
\begin{array}{cl}
   \frac{|h_{i}| \bar{\mathcal{C}}}{\tilde{\sigma}_v \sqrt{\xi_i c_0} + |h_{i}| c_0} , & ~~~~~~~i = 1  \\
0, &~~~~~~~i > 1,
\end{array} \right. \nonumber \\
g_{p,i}^{(\emph{opt},2)} ~=~& \left\{
\begin{array}{cl}
 \Big(\frac{\tilde{\sigma}_v  \bar{\mathcal{C}}}{\tilde{\sigma}_v  \xi_i +|h_{i}| \sqrt{\xi_i c_0} }\Big)^{1/2}, & ~i = 1  \\
0, &~i > 1.
\end{array} \right.
\end{flalign}
\end{theorem} 

\begin{proof}
Please see the Appendix C.
\end{proof}
Given the optimal solution of $(\boldsymbol{\kappa}, \boldsymbol{g})$, we see that the optimal error probability in (\ref{eq:pro_formu_relax}) is %\footnote{For simplicity, we neglect the rounding effect of $\boldsymbol{\kappa}$ throughout the paper.} 
\[
 \textrm{P}_{e}^{(\textrm{opt}, 2)} = Q \left( \frac{ \sqrt{\bar{\mathcal{C}}}}{2} \max \left\{ \frac{\gamma_i |h_{i}|}{\tilde{\sigma}_v \sqrt{\xi_i} + |h_{i}|\sqrt{c_0}}\right\} \right).
\] 
Since (\ref{eq:pro_formu_relax}) is the relaxation of the MINLP (\ref{eq:pro_formu}), we see that $\textrm{P}_{e}^{(\textrm{opt}, 1)} \geq \textrm{P}_{e}^{(\textrm{opt}, 2)}$ \cite{IP_Book}. In practice, we may consider a floor operation for the number of samples as a suboptimal solution for (\ref{eq:pro_formu}), i.e.,
 \begin{flalign}\label{eq:pro_formu_sub_2}
 \kappa_{p,i}^{(\textrm{sub})} = \left \lfloor \kappa_{p,i}^{(\textrm{opt}, 2)} \right\rfloor ~~~\textrm{and}~~~g_{p,i}^{(\textrm{sub})} =   g_{p,i}^{(\textrm{opt},2)},~\forall i.
 \end{flalign}
Let us denote the resulting error probability as $\textrm{P}_{e}^{(\textrm{sub})}$. 
Then we see that $\textrm{P}_{e}^{(\textrm{opt}, 2)} \leq \textrm{P}_{e}^{(\textrm{opt}, 1)} \leq \textrm{P}_{e}^{(\textrm{sub})}$. Furthermore, when $\kappa_{p,1}^{(\textrm{opt}, 2)}$ is large, based on the first-order Taylor series, we have
 \begin{flalign*}
\textrm{P}_{e}^{(\textrm{sub})} - \textrm{P}_{e}^{(\textrm{opt},2)} ~=~\textrm{P}_{e}(\boldsymbol{\kappa}-\Delta \boldsymbol{\kappa}, {\boldsymbol{g}}) - \textrm{P}_{e}(\boldsymbol{\kappa}, {\boldsymbol{g}})~\approx~ \frac{\Delta \kappa_1 \delta_0 \delta_1}{8\sqrt{2\pi} } \exp(-\delta_0^2/8) (\kappa_1 + \delta_1)^{-2} ~\rightarrow~ 0^+,
 \end{flalign*}
where $\delta_0 = g_1 \gamma_1 |h_1|/\tilde{\sigma}_v$ and $\delta_1 = g_1^2 |h_1|^2/\tilde{\sigma}_v^2$. With small value of $\Delta \kappa_1$ (normally $\Delta \kappa_1 < 1$), it is interesting to note that our rounding algorithm is near optimal with large system level cost constraint. When $\bar{\mathcal{C}}$ is relatively small, as we will show in our simulations, our proposed suboptimal algorithm can also provide a good approximation to the optimal solution.

Based on (\ref{eq:pro_formu_sub_2}), when we jointly design the number of samples and amplifier gains subject to the system level cost constraint, only {\em one} secondary user needs to be active in the cognitive radio network, i.e., collecting local energy samples and transmitting the energy statistic to the fusion center. It is interesting to note that this strategy is similar to multiuser diversity where the base station selects the user with the highest channel to achieve maximum sum rate capacity \cite{Knopp_1995}. In this case, the fusion center will select the secondary user with the largest $\rho_i$ to perform local spectrum sensing and data forwarding. This will significantly reduce the bandwidth cost for data forwarding.

\emph{Remark}: We note that the result in (\ref{eq:pro_formu_sub_2}) can be implemented in a distributed fashion. The idea is based on opportunistic carrier sensing \cite{Zhao_2005} or opportunistic relaying \cite{Blet_2006} in which a backoff timer is set to be a decreasing function of channel state information. In particular, at the beginning of each sensing time slot, the fusion center broadcasts a beacon signal to synchronize all secondary users in the cognitive radio network. After estimating the channel gain\footnote{We assume reciprocity of the uplink and downlink channels between the fusion center and secondary users \cite{Wireless_Rapp}.} $|h_i|$, the secondary user calculates the control parameter $\rho_i$ based on its local received SNR $\gamma_i$ and then maps $\rho_i$ to a backoff timer $f(\rho_i)$ (equal to $c/\rho_i$ in \cite{Blet_2006}, where $c$ is a constant). Under a collision free situation, the secondary user with largest $\rho_i$ will expire first and perform local energy calculation and data forwarding during this time slot\footnote{Detailed analysis on how to reduce the collision probability for this scheme can be found in \cite{Zhao_2005}.}. Note that in this case, fusion center does not need to broadcast the optimal design parameter for each secondary user and this will reduce the cooperative sensing cost for broadcasting and reporting.
 
 %Here, $\Delta_1 =  \textrm{N}c_0 + \boldsymbol{1}^{\texttt{T}}\boldsymbol{\xi}$.
 
%-----------------------------------------
\subsubsection{Scenario B}\label{sec:caseIII_scenarioII}
%-----------------------------------------

We examine the optimization (\ref{eq:pro_formu_ind}) over both $\boldsymbol{\kappa}$ and $ \boldsymbol{g}$. Similar to Scenario A, we first consider the relaxation to the original MINLP in (\ref{eq:pro_formu_ind}), i.e.,
\begin{flalign}\label{eq:pro_formu_g_k_scenarioII}
\underset{\boldsymbol{\kappa}, \boldsymbol{g}}{\min}  ~~~& \textrm{P}_{e}(\boldsymbol{\kappa}, \boldsymbol{g})
\nonumber  \\  \textrm{s.t.}  ~~~& \mathcal{C}(\boldsymbol{\kappa}, \boldsymbol{g}) \leq \bar{\mathcal{C}},~\boldsymbol{\kappa}\in \mathcal{R}_{+}^n, ~\boldsymbol{g}\in \mathcal{R}_{+}^n,  \nonumber
  \\   ~~~&   \boldsymbol{\kappa}\preceq {\kappa}_{\max}\boldsymbol{1},~\xi_i g_i^2 \leq \mathcal{P}_{\max}.
\end{flalign}

Again, we see that this is a convex optimization problem and can be solved by standard methods. Let us denote the optimal solution in (\ref{eq:pro_formu_g_k_scenarioII}) as $\big(\kappa_{p,i}^{(\textrm{opt})},g_{p,i}^{(\textrm{opt})}\big)$. Similarly, we note that 
\begin{lemma}\label{lemma:u_v_0_same_ScenarioII}
Optimal solution of $(\boldsymbol{\kappa}, \boldsymbol{g})$ in (\ref{eq:pro_formu_g_k_scenarioII}) should satisfy either 1) $\kappa_i > 0 $ and $g_i > 0 $, or 2) $\kappa_i = 0 $ and $g_i = 0 $ for secondary user $i$. 
\end{lemma} 

The proof is similar to that of Lemma \ref{lemma:u_v_0_same} and thus omitted. With the additional constraints imposed on $\boldsymbol{\kappa}$ and $\boldsymbol{g}$, we see that in general it is difficult to obtain the closed-form solutions for $(\boldsymbol{\kappa}, \boldsymbol{g})$. Since the optimal solution of $(\boldsymbol{\kappa}, \boldsymbol{g})$ needs to be equal to 0 or greater than 0 simultaneously, we propose a heuristic suboptimal algorithm for Scenario B. Specifically, first we assign $\kappa_{\max}$
and $\mathcal{P}_{\max}$ to the secondary user with largest $\rho_i$. If there are remaining resources, we assign $\kappa_{\max}$ and $\mathcal{P}_{\max}$ to the secondary
user with second largest $\rho_i$ and so on until $\kappa_{\max}$ and $\mathcal{P}_{\max}$ cannot be assigned to any one
secondary user. In this case, we merely utilize the near-optimal
solution in (\ref{eq:pro_formu_sub_2}) to allocate $(\kappa_i, g_i)$ to the secondary user with the next largest $\rho_i$ and $\kappa_i = 0,~ g_i = 0$ to the rest of the secondary users. 
Let us denote the suboptimal solution as $\big(\kappa_{p,i}^{(\textrm{sub})},g_{p,i}^{(\textrm{sub})}\big) $. The detailed algorithm for Scenario B is illustrated in Algorithm 1.
 
\begin{algorithm}
\caption{\emph{Heuristic Suboptimal Algorithm}}
\begin{algorithmic} 
\STATE {Sort $\rho_{i}$ in a decreasing order.}
\FOR{$i=1$ to $n$}
\IF{$c_0\kappa_{\max} + \mathcal{P}_{\max} <  \bar{\mathcal{C}}$}
\STATE {$\bar{\mathcal{C}} \leftarrow \bar{\mathcal{C}} - c_0\kappa_{\max}-\mathcal{P}_{\max} $; $\kappa_i \leftarrow \kappa_{\max}$; $g_i \leftarrow \sqrt{\mathcal{P}_{\max}/\xi_i}$.}
\ELSE
\STATE {Compute $\kappa_i$ and $g_i$ from (\ref{eq:pro_formu_sub_2})};
%\IF{$\kappa_i > \kappa_{\max}$}
%\STATE {$\kappa_i \leftarrow \kappa_{\max}$; $z_i \leftarrow (\bar{\mathcal{C}} - c_0 \kappa_{\max})/\xi_i$.} 
%\ENDIF
%\IF {$z_i > \mathcal{P}_{\max}/\xi_i$}
%\STATE {$z_i \leftarrow \mathcal{P}_{\max}/\xi_i$; $\kappa_i \leftarrow (\bar{\mathcal{C}} -  \mathcal{P}_{\max})/c_0$.} 
%\ENDIF
\STATE {Adjust and truncate $\kappa_i$ and $g_i$ to guarantee $\kappa_i \in (0, \kappa_{\max}]$ and $g_i \in\Big(0, \sqrt{\mathcal{P}_{\max}/\xi_i}\Big]$} and stop.
\ENDIF
\ENDFOR
\end{algorithmic}
\end{algorithm}

%-----------------------------------------
\subsection{Case II: Optimization of Either $\boldsymbol{g}$ or $\boldsymbol{\kappa}$} 
%-----------------------------------------

In some applications, either $\boldsymbol{g}$ or $\boldsymbol{\kappa}$ may be fixed for secondary users. For example, local energy calculation may be scheduled in a fixed time slot and each secondary user is assigned same number of samples. In this case, we need to optimize the amplifier gain to achieve the desired error probability. On the other hand, we may need to choose appropriate number of samples when the amplifier gains are fixed. Here, we first assume fixed number of samples, i.e., $\boldsymbol{\kappa} = \tilde{\boldsymbol{\kappa}}$, then we need to minimize the error probability by choosing appropriate $\boldsymbol{g}$. Let us define global transmission power constraint as $\mathcal{P}_{\textrm{tot}} = \bar{\mathcal{C}} - c_0 \boldsymbol{1}^\texttt{T} \tilde{\boldsymbol{\kappa}}$. We now examine both these cases.

%-----------------------------------------
\subsubsection{Scenario A}\label{sec:caseI}
%-----------------------------------------

Here, we minimize global error probability assuming the global transmit power constraint is given as $\mathcal{P}_{\textrm{tot}}$. We
define $z_i = g_i^2$, $a_i = \tilde{\kappa}_i \gamma_i^2$ and $b_i = \tilde{\kappa}_i\tilde{\sigma}_v^2/ |h_i|^2$. Then, the optimization problem in (\ref{eq:pro_formu}) is equivalent to 
 \begin{flalign}\label{eq:pro_formu_k_2}
\underset{\boldsymbol{z}}{\min}  ~~~& \sum_{i=1}^{n} \frac{a_i b_i}{z_i +b_i}
\nonumber  \\  \textrm{s.t.}    ~~~&   \boldsymbol{\xi}^\texttt{T}\boldsymbol{z} \leq \mathcal{P}_{\textrm{tot}},~\boldsymbol{z}\succeq \boldsymbol{0}.
\end{flalign}

%Then, the Lagrangian function can be given as
%\[
%\mathcal{L}(\boldsymbol{\kappa}, \lambda_0, \boldsymbol{u}) = \sum_{i=1}^{n} \frac{a_i b_i}{z_i +b_i}+ \lambda_0 ( \textbf{1}^\texttt{T}\boldsymbol{\kappa}  - \mathcal{P}_{\textrm{tot}} ) - \boldsymbol{u}^\texttt{T}\boldsymbol{\kappa},
%\]
%where $\lambda_0$ and $u_i  \geq 0$ are Lagrangian multipliers. 
It is easy to see that (\ref{eq:pro_formu_k_2}) is a convex optimization problem. After some manipulations, we see that the Karush-Kuhn-Tucker (KKT) conditions can be given as
 \begin{flalign}%\label{eq:kkt_g}
\textstyle  \frac{a_i b_i}{( z_i +b_i)^2}    + u_i  - \lambda_0 \xi_i~=~& 0 \label{eq:kkt_k1} \\  
\lambda_0(\boldsymbol{\xi}^\texttt{T}\boldsymbol{z} - \mathcal{P}_{\textrm{tot}}) ~=~& 0 \label{eq:kkt_k2}  \\  
u_i z_i ~=~& 0\label{eq:kkt_k3}. %\\ 
%\lambda_0~\geq~ 0,~\lambda_i ~\geq~ 0,~z_i ~\geq~& 0 \label{eq:kkt_k4}.
\end{flalign}
where $\lambda_0\geq 0$ and $u_i  \geq 0$ are Lagrangian multipliers. First we assume that $\lambda_0 > 0$ and $u_i = 0$, then from (\ref{eq:kkt_k1}), we see that
 \begin{flalign}\label{eq:k_temp}
z_i =  \Big[ \sqrt{{a_i b_i}/{(\xi_i \lambda_0)}}   - b_i\Big]^+,  
\end{flalign}
where $[x]^+ = \max\{0,x\}$. Plugging this into (\ref{eq:kkt_k2}), we have
 $
 \sqrt{\lambda_0}= \frac{\sum_{i  \in \mathcal{S}_0} \sqrt{a_i b_i \xi_i}}{\mathcal{P}_{\textrm{tot}} + \sum_{i  \in \mathcal{S}_0}b_i \xi_i}. 
$ 
where $\mathcal{S}_0 = \{i|z_i > 0\}$. Then, we need to determine the set $\mathcal{S}_0$ to obtain the closed-form solution for $\boldsymbol{z}$. To do this, let us define 
$\beta_i =\sqrt{{b_i\xi_i}/{a_i}}.$ 
Without loss of generality, we assume $\beta_1 \leq \beta_2 \leq \cdots \leq \beta_{n}$. 
%After some derivations (follow the analysis in \cite{Cui_2007}), we have 
After some derivations, as outlined in Appendix D, we have
%Following \cite{Cui_2007}, we have
\begin{equation}\label{eq:set_S0}
\mathcal{S}_0 = \left\{
\begin{array}{ll}
 \{1,\cdots,i_\mathcal{S}|f(i_\mathcal{S}) < 1,f(i_\mathcal{S}+1) \geq 1 \}, &   f(n) \geq 1  \\
\{1,\cdots,n \}, &    \textrm{otherwise},
\end{array} \right.
\end{equation}
where
\begin{equation}\label{eq:compute_f_1st}
 f(i) ~=~  \textstyle \frac{\beta_i\sum_{j =1}^i \sqrt{a_j b_j \xi_j}}{\mathcal{P}_{\textrm{tot}}+ \sum_{j =1}^i b_j \xi_j }.   
\end{equation}
Thus, plugging $\lambda_0$ into (\ref{eq:k_temp}), the optimal amplifier gains can be obtained as
\begin{equation}\label{eq:opti_k_sol}
g_{p,i}^{(\textrm{opt})} = \left\{
\begin{array}{cl}
 \Big[\frac{{\tilde{\kappa}_i\tilde{\sigma}_v^2}}{|h_i|^2 } \Big( \frac{\gamma_i|h_i|}{\sqrt{\xi_i}} \eta - 1 \Big)\Big]^{1/2}, & i \in \mathcal{S}_0  \\
0, &i \notin \mathcal{S}_0,
\end{array} \right.
\end{equation}
where 
$
\eta = \frac{\sum_{i  \in \mathcal{S}_0} \tilde{\kappa}_i\xi_i/ |h_i|^2 + \mathcal{P}_{\textrm{tot}}/ \tilde{\sigma}_v^2 }{\sum_{i  \in \mathcal{S}_0} \tilde{\kappa}_i\sqrt{\xi_i}\gamma_i/|h_i|}.
$

 \emph{Remark}:
The optimal amplifier gains follow the water-filling strategy, i.e., with larger $\beta_i$, the chance for the secondary user to be inactive is higher, where $\beta_i$ is a measure of the observation and fusion channel quality. Note that 
$
\beta_i \propto {1}/({\gamma_i|h_i|}).
$
Hence, when the local received SNR is low or the fusion channel quality is poor, the secondary user tends not to transmit the local calculated energy to the fusion center.  

For comparison, we consider two suboptimal solutions for this optimization problem:
1) A simple solution is to choose equal transmission power for each secondary user, i.e.,
$
g_{p,i}^{(\textrm{equ})} =  \sqrt{\mathcal{P}_{\textrm{tot}}/(n\xi_i)}$; 
2) Using the Cauchy-Schwarz inequality, we see that $\textrm{P}_{e}({\kappa}_{\infty})$ in (\ref{eq:asym_kappa}) can be minimized when $g_i = c \gamma_i^2|h_i|^2/\xi_i^2$, where $c$ is a constant. Based on this, we propose an alternate suboptimal solution for amplifier gains, i.e.,
$
\textstyle g_{p,i}^{(\textrm{sub})} = \left( \frac{\gamma_i^2|h_i|^2/\xi_i^2}{\sum_{i=1}^{n}  \gamma_i^2|h_i|^2/\xi_i}\mathcal{P}_{\textrm{tot}} \right)^{1/2}.
$ Let us denote the asymptotic detection probability when $\tilde{\kappa}_i  \rightarrow \infty$ for these three solutions of amplifier gains as $\textrm{P}_{e}^{(\textrm{opt})}({\kappa}_{\infty})$, $\textrm{P}_{e}^{(\textrm{equ})}({\kappa}_{\infty})$ and $\textrm{P}_{e}^{(\textrm{sub})}({\kappa}_{\infty})$. Then, we note that
\begin{lemma}\label{lemma:asym_k_diff}
When $\beta_2 > \beta_1$, $\emph{P}_{e}^{(\emph{equ})}({\kappa}_{\infty}) \geq  \emph{P}_{e}^{(\emph{sub})}({\kappa}_{\infty}) \geq \emph{P}_{e}^{(\emph{opt})}({\kappa}_{\infty})$.
\end{lemma} 

\begin{proof}
Please see the Appendix E.
\end{proof}

%-----------------------------------------
\subsubsection{Scenario B}\label{sec:caseI_ScenarioII}
%-----------------------------------------

Next, we minimize global error probability assuming the global transmit power constraint $\mathcal{P}_{\textrm{tot}}$ and the individual transmit power limit $\mathcal{P}_{\max}$. In this scenario, the optimization problem in (\ref{eq:pro_formu_ind}) becomes 
 \begin{flalign}\label{eq:pro_formu_k_ind}
\underset{\boldsymbol{z}}{\min}  ~~~& \sum_{i=1}^{n} \frac{a_i b_i}{z_i +b_i}
\nonumber  \\  \textrm{s.t.}  ~~~&  \boldsymbol{\xi}^{\texttt{T}} \boldsymbol{z} \leq \mathcal{P}_{\textrm{tot}},~\boldsymbol{z}\succeq \boldsymbol{0},~\xi_i z_i \leq \mathcal{P}_{\max}.
\end{flalign}
With the additional constraint in (\ref{eq:pro_formu_k_ind}) as compared to (\ref{eq:pro_formu_k_2}), the updated KKT conditions are
 \begin{flalign}%\label{eq:kkt_g}
\textstyle  \frac{a_i b_i}{( z_i +b_i)^2}    + u_i -v_i \xi_i  - \lambda_0 \xi_i~=~& 0 \label{eq:kkt_k1_ind} \\  
v_i (\xi_i z_i-\mathcal{P}_{\max}) ~=~& 0\label{eq:kkt_k2_ind}, %\\ 
%\lambda_0~\geq~ 0,~\lambda_i ~\geq~ 0,~\kappa_i ~\geq~& 0 \label{eq:kkt_k4}.
\end{flalign}
where $v_i  \geq 0$ are Lagrangian multipliers. First we assume that $\lambda_0 > 0$ and $u_i = v_i = 0$, then from (\ref{eq:kkt_k1_ind}), we see that 
$
z_i =   \sqrt{{a_i b_i}/{(\xi_i \lambda_0)}}   - b_i .
$
%To satisfy $0 < \kappa_i < \kappa_{\max}$, we have
%\[
%\textstyle \frac{\sqrt{a_i b_i}}{\kappa_{\max} + b_i} < \sqrt{\lambda_0} <\sqrt{\frac{a_i}{b_i}}.
%\]
Thus, based on the value of $\sqrt{\lambda_0}$, we can determine the optimal solution of $z_i$ as
\begin{equation*}
z_i = \left\{
\begin{array}{ll}
 0, & \textrm{if}~\sqrt{\lambda_0} > \sqrt{a_i/(b_i \xi_i)}  \\
\mathcal{P}_{\max}/\xi_i, & \textrm{if}~ 0<\sqrt{\lambda_0} < \sqrt{a_i b_i \xi_i}/(\mathcal{P}_{\max} + b_i \xi_i)  \\
 \sqrt{{a_i b_i}/{(\xi_i \lambda_0)}}   - b_i, &    \textrm{otherwise}.
\end{array} \right.
\end{equation*}
Let us define two disjoint sets for secondary users as $\mathcal{S}_1 = \{i|z_i = \mathcal{P}_{\max}/\xi_i\}$ and $\mathcal{S}_2 = \{i|0< z_i < \mathcal{P}_{\max}/\xi_i\}$. Plugging $z_i$ into (\ref{eq:kkt_k2}), we have
\[
\textstyle |\mathcal{S}_1|\mathcal{P}_{\max} + \left(1/\sqrt{\lambda_0}\,\right)\sum_{i  \in \mathcal{S}_2} \sqrt{a_i b_i \xi_i} - \sum_{i  \in \mathcal{S}_2} b_i\xi_i ~=~ \mathcal{P}_{\textrm{tot}}, 
\]
which implies that
$
  \sqrt{\lambda_0}  =\frac{\sum_{i  \in \mathcal{S}_2} \sqrt{a_i b_i \xi_i}}{\mathcal{P}_{\textrm{tot}} -|\mathcal{S}_1|\mathcal{P}_{\max} +\sum_{i  \in \mathcal{S}_2} b_i \xi_i}.
$

In order to determine $\mathcal{S}_1$, $\mathcal{S}_2$ and $\sqrt{\lambda_0}$ and thus obtain the closed-form solution for $z_i$, we propose a two-stage generalized water-filling algorithm as follows:
\begin{enumerate}
	\item In the first stage, we aim to determine the set $\mathcal{S}_1$. To do this, let us define 
$\tilde{\beta}_{i} =\frac{\mathcal{P}_{\max} + b_i \xi_i}{\sqrt{a_i b_i \xi_i}}.$ 
Without loss of generality, we assume $\tilde{\beta}_1 \leq \tilde{\beta}_2 \leq \cdots \leq \tilde{\beta}_{n}$. Then, similar to Scenario A, $\mathcal{S}_1$ can be obtained by (\ref{eq:set_S0}) with 
\begin{equation}\label{eq:compute_f_sec}
  \tilde{f}(i) ~=~  \frac{\tilde{\beta}_{i}\sum_{m \in \tilde{\mathcal{S}}_i} \sqrt{a_m b_m \xi_m}}{\mathcal{P}_{\textrm{tot}} - i \mathcal{P}_{\max} + \sum_{m \in \tilde{\mathcal{S}}_i} b_m \xi_m},~~i \leq \bigg \lfloor\frac{\mathcal{P}_{\textrm{tot}}}{\mathcal{P}_{\max}} \bigg\rfloor,
\end{equation}
where $\tilde{\mathcal{S}}_i = \{m|\beta_m < \tilde{\beta}_{i},~i<m\leq n\}$. For an outline, please see Appendix F.
After $\mathcal{S}_1$ is determined, we have $z_i = \mathcal{P}_{\max}/\xi_i,~\forall i\in \mathcal{S}_1$.
\item In the second stage, we follow the similar procedure in Scenario A to obtain $\mathcal{S}_2$ and $z_i$ for $i \notin \mathcal{S}_1$. The solution is given in (\ref{eq:opti_k_sol}), except that $\mathcal{P}_{\textrm{tot}}$ and $n$ are replaced by $\mathcal{P}_{\textrm{tot}} - |\mathcal{S}_1|\mathcal{P}_{\max}$ and $n - |\mathcal{S}_1|$, respectively.
\end{enumerate}

To summarize, the detailed generalized water-filling algorithm for Scenario B is illustrated in Algorithm 2. 
With amplifier gains fixed, we need to optimize the number of samples to achieve the desired error probability. In this case, the solutions of the number of samples are similar to those of the amplifier gains in both scenarios (with additional relaxation consideration), thus omitted from this paper. %For instance, in Scenario A, the optimal number of samples for each secondary user follows a water-filling strategy. %Specifically, when $\gamma_i$ is low, secondary user $i$ tends 
%not to collect local energy samples. 
%Specifically, when the local received SNR is low, the secondary user tends not to collect local energy samples. 

\begin{algorithm}
\caption{\emph{Generalized Water-filing Algorithm}}
\begin{algorithmic}
\STATE \underline{\textbf{Stage 1}}: {Sort $\tilde{\beta}_{i}$ in an increasing order.}
\FOR{$i=1$ to $\big \lfloor\frac{\mathcal{P}_{\textrm{tot}}}{\mathcal{P}_{\max}} \big\rfloor$}
\STATE {Compute $\tilde{f}(i)$ from (\ref{eq:compute_f_sec})};
\IF{$\tilde{f}(i) \geq 1$}
\STATE {Set $\mathcal{S}_1 = \{1,\cdots, i\}$ and stop.}
\ENDIF
\ENDFOR
\FOR{$i \in \mathcal{S}_1$} \STATE $z_i \leftarrow \mathcal{P}_{\max}/\xi_i $.
\ENDFOR
\STATE \underline{\textbf{Stage 2}}: {For $j\notin \mathcal{S}_1$, sort $\beta_j$ in an increasing order and set $\mathcal{P}_{\textrm{tot}} \leftarrow \mathcal{P}_{\textrm{tot}} - |\mathcal{S}_1|\mathcal{P}_{\max}$ and $n \leftarrow n - |\mathcal{S}_1|$.}
\FOR{$j=1$ to $n$}
\STATE Compute $f(j)$ from (\ref{eq:compute_f_1st});
\IF{$f(j) \geq 1$}
\STATE Set $\mathcal{S}_2 = \{1,\cdots, j\}$ and stop. 
\ENDIF
\ENDFOR
\FOR{$j \in \mathcal{S}_2$} 
\STATE Compute $\eta$ and $z_j$ from (\ref{eq:opti_k_sol}).
\ENDFOR
\end{algorithmic}
\end{algorithm}
%\vspace{-0.35in}

 %-----------------------------------------
\section{Optimization: Minimization of System Level Cost}
%-----------------------------------------

In the section, we aim to minimize the system level cost of cooperative spectrum sensing to achieve a
targeted error probability. Similar to the optimization problem in Section \ref{sec:primal}, we consider two scenarios
which depend on whether additional constraints are imposed or not. For instance, in Scenario
A, the optimization problem can be formulated as:
  \begin{flalign}\label{eq:cost_min}
\underset{\boldsymbol{\kappa}, \boldsymbol{g}}{\min}  ~~~& \mathcal{C}(\boldsymbol{\kappa}, \boldsymbol{g})
\nonumber  \\  \textrm{s.t.}  ~~~&  \textrm{P}_e(\boldsymbol{\kappa}, \boldsymbol{g}) \leq \bar{\textrm{P}}_e, ~\boldsymbol{\kappa}\in \mathcal{Z}_{+}^n, ~\boldsymbol{g}\in \mathcal{R}_{+}^n,
\end{flalign}
where $\bar{\textrm{P}}_e$ is a predefined error probability threshold. Similar to the analysis in Section \ref{sec:caseIII_scenarioI}, we consider the relaxation, i.e., $\boldsymbol{\kappa}\in \mathcal{R}_{+}^n$ to this MINLP, and the optimal solution of this relaxation problem is stated as follows:
\begin{theorem}\label{thm:caseIII_sol_dual}
Consider the optimization problem in (\ref{eq:cost_min}) and $\rho_i$ as defined in Theorem \ref{thm:caseIII_sol}. Then, 
\begin{flalign}%\label{eq:sol_dual_formu_g_k}
\kappa_{d,i}^{(\emph{opt})} =& \left\{
\begin{array}{cl}
   \frac{\epsilon}{\gamma_i^2} \Big(1 + \sqrt{\frac{\xi_i}{c_0}}\frac{\tilde{\sigma}_v}{|h_i|} \Big) , & ~~~~~~~~~~~i = 1  \\
0, &~~~~~~~~~~~i > 1,
\end{array} \right. \nonumber \\
g_{d,i}^{(\emph{opt})} =& \left\{
\begin{array}{cl}
 \Big[ \frac{\epsilon \tilde{\sigma}_v^2}{\gamma_i^2|h_i|^2 } \Big(1 + \sqrt{\frac{c_0}{\xi_i}} \frac{|h_i|}{\tilde{\sigma}_v}\Big)\Big]^{1/2}, & ~ i = 1  \\
0, & ~i > 1,
\end{array} \right.
\end{flalign}
where $\epsilon = 4 [Q^{-1}(\bar{\emph{P}}_e)]^2$.
\end{theorem}

The proof is similar to that of Theorem \ref{thm:caseIII_sol} and thus omitted. Similarly, we may consider a ceiling operation for the number of samples as a near-optimal solution for (\ref{eq:cost_min}). Additionally, we see that only {\em one} secondary user needs be active for collecting the samples for local energy calculation and transmitting energy statistics to fusion center. We have separately examined the optimization problem for the remaining cases considered in Section \ref{sec:primal}, i.e., when jointly designing $\boldsymbol{\kappa}$ and $ \boldsymbol{g}$ for Scenario B; and when designing either $\boldsymbol{\kappa}$ or $ \boldsymbol{g}$ for both Scenarios A and Scenario B. Due to space limitations, we omit the discussions in the paper.

%\vspace{-0.15in}                                                                                                    
%-----------------------------------------
\section{Simulation Results}
%-----------------------------------------

In this section, we present numerical results for system level performance evaluation and optimal design for cooperative spectrum sensing in cognitive radio networks. In the following results, we assume ${\sigma}_n^2 = {\sigma}_v^2 = 1$ and 
$c_0 = 1$.

%\vspace{-0.15in}
 %-----------------------------------------
\subsection{Average Error Probability}
%-----------------------------------------

In Fig.~\ref{fig:Avg_Perf}, we plot the average error probability versus the (equal) amplifier gain for all three channel scenarios from Table \ref{tb1:3_scenarios}. We see that in the low and moderate fusion SNR regimes, Channel Environment II (Rayleigh fading observation channels and
AWGN fusion channels) provides the lowest average error probability among all three scenarios. Thus, to maintain a desired detection performance, the fusion
channels need to be as reliable as possible, while the local
received SNRs can be dynamic and be used to exploit spatial diversity.
%However, in the high SNR regime, we observe that Channel Environment I (AWGN observation channels and 
%Rayleigh fading fusion channels) has the best performance. This is because in the high
%global SNR regime, fusion channel gain does not
%impact error probability and reliable observation
%channels are needed to improve detection performance. 

\begin{figure}[htb] 
 \centerline{\epsfxsize 4.5in\epsfbox{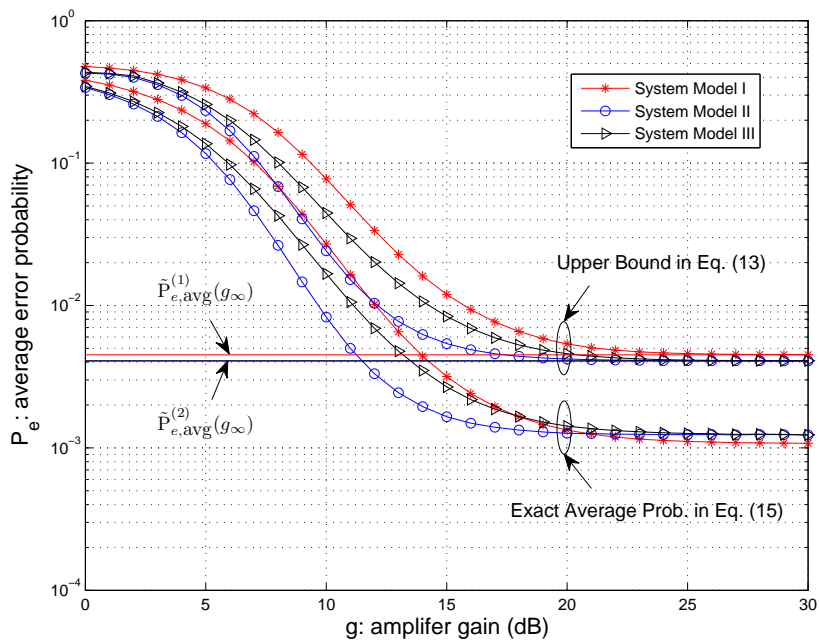}}
 \caption {Average error probability for cooperative spectrum sensing. In the simulation, we choose $\bar{\gamma} = -8\textrm{dB}$, $\kappa = 100$ and $n = 15$.}\label{fig:Avg_Perf} \vspace{-0.25in}
\end{figure}

%-----------------------------------------
\subsection{Minimization of Error Probability}
%-----------------------------------------

For the optimal system design, we assume $n = 6$, $\boldsymbol{h} = [1.56, 1.99, 0.37, 1.52, 0.39, 1.98]^{\texttt{T}}$ and $\boldsymbol{\gamma} = [-8.86,  -15.23, -7.21, -5.09, -10.00,  -10.97]^{\texttt{T}}(\textrm{dB})$. Here, we define the global fusion SNR as $\snr = \mathcal{P}_{\textrm{tot}}/\sigma_v^2$. For comparison, we consider equal number of samples and amplifier gains as a suboptimal solution. 

In Fig.~\ref{fig:Primal_caseIII}, we plot the error probability versus system level cost constraint in Case I for joint optimization of $\boldsymbol{\kappa}$ and $\boldsymbol{g}$ for both Scenario A and B. In this simulation, we utilize standard MINLP methods \cite{MINLP_Book} for optimization problem in (\ref{eq:pro_formu}); the closed-form solution $\big(\kappa_{p,i}^{(\textrm{opt}, 2)}, g_{p,i}^{(\textrm{opt}, 2)}\big)$ in Theorem \ref{thm:caseIII_sol} for the convex relaxation of the optimization problem in (\ref{eq:pro_formu_relax}); and our proposed suboptimal solution $\big(\kappa_{p,i}^{(\textrm{sub})}, g_{p,i}^{(\textrm{sub})}\big)$ in (\ref{eq:pro_formu_sub_2}) in Scenario A  and interior-point method to solve the optimization problem in Scenario B. As expected, we see that in Scenario A, the error probability of optimization problem in (\ref{eq:pro_formu}) and its relaxation in (\ref{eq:pro_formu_relax}) converges, even with relatively small system level cost constraint. Also, our proposed suboptimal solution in (\ref{eq:pro_formu_sub_2}) is near optimal as previously mentioned. Furthermore, we observe that the error performance is degraded with the additional constraints in Scenario B. Additionally, our proposed suboptimal algorithm in Scenario B has negligible performance loss compared to the optimal solution. 

\begin{figure}[htb]
 \centerline{\epsfxsize 4.28in\epsfbox{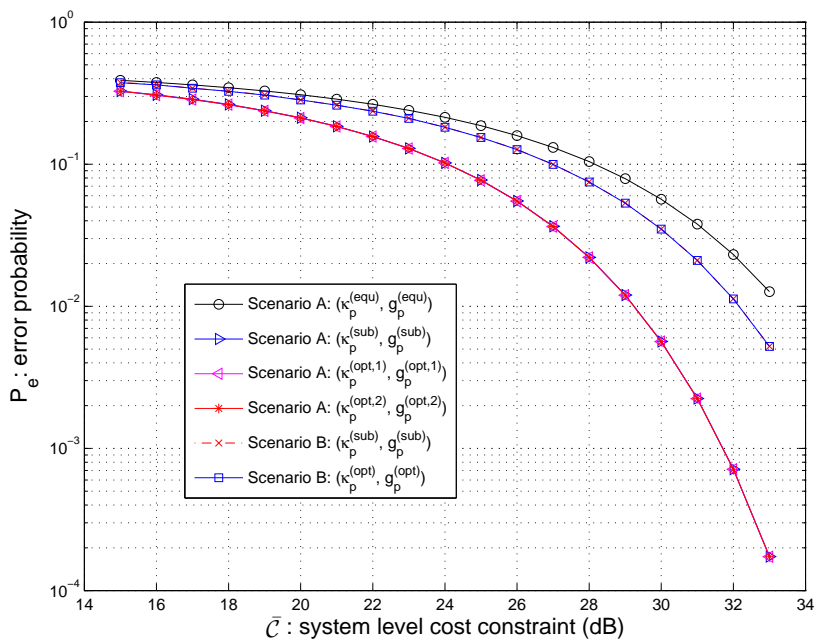}}%\vspace{-0.1in}
 \caption {Case I: error probability for different solutions of $(\boldsymbol{\kappa}, \boldsymbol{g})$. In Scenario B, we choose $\kappa_{\max} = 0.2 \lfloor\bar{\mathcal{C}}/c_0\rfloor$ and $\mathcal{P}_{\max} = 0.2\bar{\mathcal{C}}$.}\label{fig:Primal_caseIII}
\end{figure}

Fig.~\ref{fig:Primal_caseII} shows the error probability versus total number of samples in Case II (optimization of $\boldsymbol{g}$ given $\tilde{\boldsymbol{\kappa}}$). 
As expected, we see that the optimal solution provides superior performance to suboptimal solutions. From the plots, we also observe that with additional individual constraints, the optimal solution for Scenario
B performs worse than that of Scenario A. Furthermore, when total number of samples increases, we see that the error probability approaches the asymptotic bound. In particular, $\textrm{P}_{e}^{(\textrm{equ})}({\kappa}_{\infty}) \geq \textrm{P}_{e}^{(\textrm{sub})}({\kappa}_{\infty}) \geq \textrm{P}_{e}^{(\textrm{opt})}({\kappa}_{\infty})$ as stated in Lemma \ref{lemma:asym_k_diff}.

\begin{figure}[htb]
 \centerline{\epsfxsize 4.5in\epsfbox{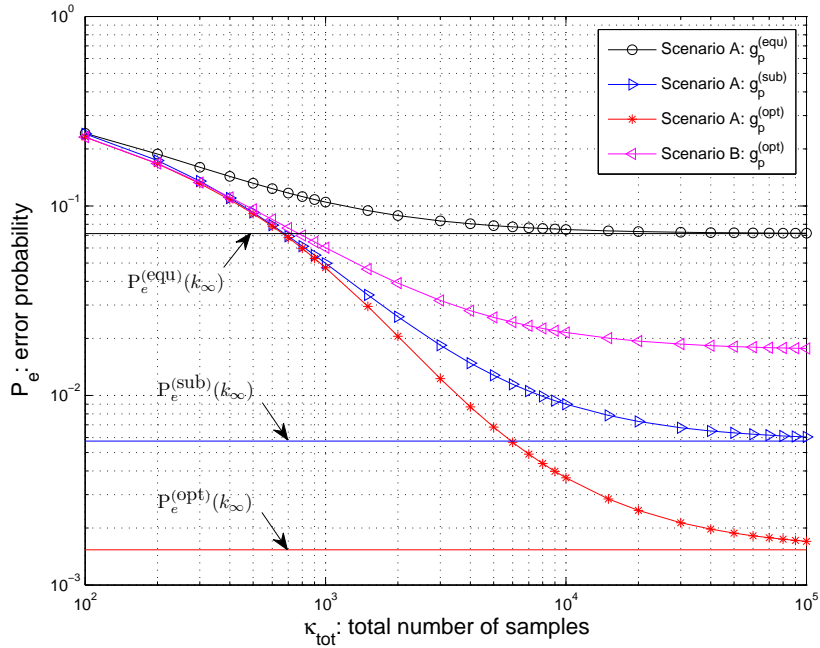}} 
 \caption {Case II: error probability for different solutions of $\boldsymbol{g}$. In the simulation, we choose $\snr = 25\textrm{dB}$ and fixed number of samples $\tilde{{\kappa}}_i = \lfloor \kappa_{\textrm{tot}}/n \rfloor$. In Scenario B, we choose  $\mathcal{P}_{\max} = 0.4\mathcal{P}_{\textrm{tot}}$.}\label{fig:Primal_caseII}
\end{figure}

%-----------------------------------------
\section{Conclusions}
%-----------------------------------------

In this paper, we present the performance evaluation and optimal design for spectrum sensing
in the cognitive radio networks. We first analyze the average error probability by considering
a range of channel realizations between the primary user and the secondary users and
between the secondary users and the fusion center. Then, we investigate the optimization problems for spectrum sensing. In particular, when jointly designing the number of samples and amplifier gains, we demonstrate that only {\em one} secondary
user needs be active, i.e., collecting local energy samples and transmitting energy statistic to
fusion center. Furthermore, we derive closed-form expressions for
optimal solutions and propose a generalized water-filling algorithm when number of samples
or amplifier gains are fixed and additional constraints are imposed.

%\vspace{-0.15in}
%%-----------------------------------------
\section{Appendix}
%%-----------------------------------------
 
%-----------------------------------------
\subsection{Proof of Convexity of Optimization Problem (\ref{eq:pro_formu_relax})} \label{sec:apped1}
%-----------------------------------------
 \begin{proof}
Let us define $z_i = g_i^2$, $p_i = \tilde{\sigma}_v^2/(\gamma_i^2 |h_{i}|^2)$, $q_i = 1/\gamma_i^2$ and $
\mathcal{F}_i(\kappa_i, z_i) = \frac{\kappa_i z_i}{p_i \kappa_i + q_i z_i}
$. To simplify our analysis, when $\kappa_i = z_i = 0$, we assume $\mathcal{F}_i(\kappa_i, z_i) = 0$\footnote{In practice, this assumption can be alleviated by adding a sufficiently small constant in the denominator.}.  Then, the optimization problem (\ref{eq:pro_formu_relax}) becomes
\begin{flalign}\label{eq:pro_formu_g_k_1}
\underset{\boldsymbol{\kappa}, \boldsymbol{z}}{\max}  ~~~&    \sum_{i=1}^{n} \mathcal{F}_i(\kappa_i, z_i)
\nonumber  \\  \textrm{s.t.}  ~~~&   c_0 \textbf{1}^\texttt{T}\boldsymbol{\kappa}  + \boldsymbol{\xi}^{\texttt{T}}\boldsymbol{z} \leq   \bar{\mathcal{C}}, ~  \boldsymbol{\kappa}\succeq  \boldsymbol{0},~\boldsymbol{z}\succeq  \boldsymbol{0}.
\end{flalign} 
After some manipulations, we see that the Hessian of $\mathcal{F}_i(\kappa_i, z_i)$ is given as
\[
\nabla^2 \mathcal{F}_i(\kappa_i, z_i) ~=~ -\frac{2p_i q_i}{(p_i \kappa_i + q_i z_i)^3} 
\begin{bmatrix}
z_i \\ \kappa_i
\end{bmatrix}
\begin{bmatrix}
z_i \\ \kappa_i
\end{bmatrix}^{\texttt{T}} \preceq  \boldsymbol{0}.
\] 
Thus, $\mathcal{F}_i(\kappa_i, z_i)$ is a concave function, which indicates that the objective function in (\ref{eq:pro_formu_g_k_1}) is also concave. This completes the proof.
\end{proof}
%\vspace{-0.15in}
%-----------------------------------------
\subsection{Proof of Lemma \ref{lemma:u_v_0_same}}
%-----------------------------------------
 \begin{proof}
 We prove this lemma by contradiction. First we assume that $(\boldsymbol{\kappa}, \boldsymbol{z})$ with $\kappa_i =0, z_i >0$ or $\kappa_i >0, z_i =0$ for secondary user $i$ is the optimal solution for (\ref{eq:pro_formu_g_k_1}). Let us define the optimal value is $p^*$. Since $\kappa_i z_i = 0$, the objective function remains unchanged in (\ref{eq:pro_formu_g_k_1}). Then, the optimization problem becomes 
 \begin{flalign}\label{eq:proof_k_z_1}
\underset{\boldsymbol{\kappa}, \boldsymbol{z}}{\max}  ~~~&  \textstyle  \sum_{j=1, j\neq i}^{n}\mathcal{F}_j(\kappa_j, z_j)
\nonumber  \\  \textrm{s.t.}  ~~~&   \textstyle c_0 \sum_{j=1, j\neq i}^{n}  \kappa_j + \sum_{j=1, j\neq i}^{n} \xi_j z_j\leq   \bar{\mathcal{C}}' \nonumber \\   ~~~&  \kappa_j \geq 0,~z_j \geq 0,~~\forall j\neq i.
\end{flalign}
where $\bar{\mathcal{C}}' = \bar{\mathcal{C}} - \xi_i z_i$ when $\kappa_i =0, z_i >0$, or $\bar{\mathcal{C}}' = \bar{\mathcal{C}} - c_0 \kappa_i$ when $\kappa_i >0, z_i =0$. In either case, we see that $\bar{\mathcal{C}}' < \bar{\mathcal{C}}$.
To prove this lemma, we need to find a substitute solution $(\boldsymbol{\kappa}', \boldsymbol{z}')$ with optimal value $p'^* > p^*$. To do this, let us replace the solution for secondary user $i$ as $\kappa_i'=  z_i' =0$. In this case, the optimization problem becomes 
 \begin{flalign}\label{eq:proof_k_z_2}
\underset{\boldsymbol{\kappa}, \boldsymbol{z}}{\max}  ~~~&  \textstyle  \sum_{j=1, j\neq i}^{n} \mathcal{F}_j(\kappa_j, z_j)
\nonumber  \\  \textrm{s.t.}  ~~~&   \textstyle c_0 \sum_{j=1, j\neq i}^{n}  \kappa_j + \sum_{j=1, j\neq i}^{n} \xi_j z_j\leq   \bar{\mathcal{C}} \nonumber \\   ~~~&  \kappa_j \geq 0,~z_j \geq 0,~~\forall j\neq i.
\end{flalign}
Then, we see that it is equivalent to proving that the optimal value $p'^*$ in (\ref{eq:proof_k_z_2}) is greater than $p^*$ in (\ref{eq:proof_k_z_1}). Since the objective and constraint functions in these two optimization problems are identical, this can be easily proved by convex relaxation in optimization problem, which implies that we can find a substitute solution $(\boldsymbol{\kappa}', \boldsymbol{z}')$, i.e., $p'^* > p^*$. This contradicts the assumption that $(\boldsymbol{\kappa}, \boldsymbol{z})$ is the optimal solution and we can conclude the proof.
\end{proof}
%\vspace{-0.15in}
%-----------------------------------------
\subsection{Proof of Theorem \ref{thm:caseIII_sol}}
%-----------------------------------------

\begin{proof}
The Lagrangian function of (\ref{eq:pro_formu_g_k_1}) can be given as
\begin{flalign*}
\mathcal{L}(\boldsymbol{\kappa},\boldsymbol{z}, \lambda_0, \boldsymbol{u}, \boldsymbol{v}) ~=~ -\sum_{i=1}^{n} \frac{\kappa_i z_i}{p_i \kappa_i + q_i z_i} + \lambda_0 (c_0 \textbf{1}^\texttt{T}\boldsymbol{\kappa}  + \boldsymbol{\xi}^{\texttt{T}}\boldsymbol{z} ) - \boldsymbol{u}^{\texttt{T}}\boldsymbol{\kappa} - \boldsymbol{v}^{\texttt{T}}\boldsymbol{z} - \lambda_0  \bar{\mathcal{C}},
\end{flalign*}
where $\lambda_0\geq0$, $u_i  \geq 0$ and $v_i  \geq 0$ are Lagrangian multipliers. Here the KKT conditions are
\begin{flalign}%\label{eq:kkt_g}
\textstyle \frac{q_i z_i^2}{(p_i \kappa_i + q_i z_i)^2}    + u_i  - c_0 \lambda_0~=~& 0 \label{eq:kkt_III_k1} \\  
\textstyle \frac{p_i \kappa_i^2}{(p_i \kappa_i + q_i z_i)^2}    + v_i  - \xi_i \lambda_0~=~& 0 \label{eq:kkt_III_k2} \\ 
\lambda_0\left(c_0 \textbf{1}^\texttt{T}\boldsymbol{\kappa}  + \boldsymbol{\xi}^{\texttt{T}}\boldsymbol{z}-   \bar{\mathcal{C}}\,\right) ~=~& 0 \label{eq:kkt_III_k3}  \\  
u_i \kappa_i ~=~ 0, ~~v_i z_i ~=~& 0.\label{eq:kkt_III_k4} 
%v_i z_i ~=~& 0\label{eq:kkt_III_k5}.
\end{flalign} 
From Lemma \ref{lemma:u_v_0_same}, we see that $u_i$ and $v_i$ need to be 0 or greater than 0 simultaneously. First we assume $u_i = v_i = 0$ and $\lambda_0 >0$, which indicates that $\kappa_i >0$ and $z_i >0$. Then from (\ref{eq:kkt_III_k1}) and (\ref{eq:kkt_III_k2}), we have
%\begin{flalign}\label{eq:kappa_z_rel}
%q_i z_i^2 (\xi_i \lambda_0 - v_i) = p_i \kappa_i^2 (c_0 \lambda_0 - u_i).
%\end{flalign}
%Based on (\ref{eq:kkt_III_k4}) and (\ref{eq:kkt_III_k5}), we can further simplify (\ref{eq:kappa_z_rel}) as
$
z_i = \omega_i \kappa_i,
$ 
where $\omega_i = \sqrt{c_0 p_i/(q_i \xi_i)}$. %From (\ref{eq:k_z_relation}), we see that the optimal amplifier gain for one active secondary user is a scalar multiplication of the number of samples. 
Plugging this into  (\ref{eq:pro_formu_g_k_1}), the original optimization problem becomes
\begin{flalign}\label{eq:pro_formu_g_k_2}
\underset{\boldsymbol{\kappa} }{\max}  ~~~&    \textstyle  \sum_{i\in \mathcal{I}} s_{1i} \kappa_i
\nonumber  \\  \textrm{s.t.}  ~~~&    \textstyle   \sum_{i\in \mathcal{I}}  s_{2i} \kappa_i \leq   \bar{\mathcal{C}}, ~\kappa_i \geq 0,~\forall i\in \mathcal{I},
\end{flalign}
where $\mathcal{I} = \{i|\kappa_i > 0, z_i >0\}$, $s_{1i} = (q_i + p_i/\omega_i)^{-1}$ and $s_{2i} = c_0 + \xi_i \omega_i$. Since adding zero will not change the objective function and constraints in (\ref{eq:pro_formu_g_k_2}), we can rewrite (\ref{eq:pro_formu_g_k_2}) as
\begin{flalign}\label{eq:pro_formu_g_k_3}
\underset{\boldsymbol{\kappa} }{\max}  ~~~&    \boldsymbol{s}_1^{\texttt{T}}\boldsymbol{\kappa}
\nonumber  \\  \textrm{s.t.}  ~~~&    \boldsymbol{s}_2^{\texttt{T}}\boldsymbol{\kappa} \leq   \bar{\mathcal{C}} , ~\boldsymbol{\kappa}\succeq  \boldsymbol{0}.
\end{flalign}
This is a classic linear optimization problem; thus we can solve this easily. Since the vertices of the polyhedron are the basic feasible solution for linear optimization problem \cite{Linear_Opt}, the optimal solution of (\ref{eq:pro_formu_g_k_3}) suggests that only one of $\kappa_i$ is non-zero while others are all zero. Let us define 
$\rho_i={s_{1i}}/{s_{2i}}$ and assume $\rho_1 \geq \rho_2 \geq \cdots \geq \rho_{n}$. Then, the optimal solution of $(\boldsymbol{\kappa}, \boldsymbol{g})$ can be given in Theorem \ref{thm:caseIII_sol}. This completes the proof.
 \end{proof}

% 
%-----------------------------------------
\subsection{Solution for Set $\mathcal{S}_0$}\label{proof:solu}
%-----------------------------------------

Here we follow the analysis in \cite{Cui_2007} to find $\mathcal{S}_0$. From (\ref{eq:k_temp}), we see that in order to guarantee $z_i \geq 0$, we need to have $\sqrt{\lambda_0} \leq \sqrt{a_i/(b_i \xi_i)}$, which indicates $f(i) < 1$ for some $i$s. Then, the problem can be stated as: given $\beta_1 \leq \beta_2 \leq \cdots \leq \beta_{n}$, $f(i_\mathcal{S}) < 1$ and $f(i_\mathcal{S}+1) \geq 1$, we have
\begin{enumerate}
	\item $f(i)$ is an increasing function of $i$ for $i \leq i_{\mathcal{S}}$;
	\item $f(i) \geq 1$ for $i > i_{\mathcal{S}}$.
\end{enumerate}
\begin{proof}
It is straightforward to show that $f(1) <1$. This indicates that $\mathcal{S}_0 \neq \emptyset $ and thus there exist feasible solutions for $\boldsymbol{z}$. When $i > 1$, we have
 \begin{flalign*}
f(i+1) ~=~& \textstyle \frac{\beta_{i+1}\sum_{j =1}^{i} \sqrt{a_j b_j \xi_j} + b_{i+1}\xi_{i+1}}{\sum_{j =1}^{i}  b_j\xi_j +  \mathcal{P}_{\textrm{tot}}+ b_{i+1} \xi_{i+1}} \\
~ \geq ~&  \textstyle \frac{\beta_{i}\sum_{j =1}^{i} \sqrt{a_j b_j\xi_j} + b_{i+1} \xi_{i+1}}{\sum_{j =1}^{i}  b_j\xi_j +  \mathcal{P}_{\textrm{tot}}+ b_{i+1} \xi_{i+1}} \\
~\mathop{\geq}^{(a)} ~ & 
\left\{
\begin{array}{cl}
f(i), &   ~~i < i_{\mathcal{S}}\\
1,  &  ~~i > i_{\mathcal{S}}.
\end{array} \right.
\end{flalign*} 
The first inequality in $(a)$ is valid since when $x/y < 1$, we have $(x+c)/(y+c) \geq x/y$, where $x,y,c>0$. Then, we see that $f(i)$ is an increasing function of $i$ for $i \leq i_{\mathcal{S}}$. The second inequality in $(a)$ is valid since when $x/y \geq 1$, we have $(x+c)/(y+c) \geq 1$. This indicates that when $f(i) \geq 1$, $f(i+1) \geq 1$ for $i > i_{\mathcal{S}}$. This completes the proof.
\end{proof}
%\vspace{-0.15in}

%-----------------------------------------
\subsection{Proof of Lemma \ref{lemma:asym_k_diff}}
%-----------------------------------------

\begin{proof}
From Section \ref{sec:caseI}, we see that
\[
f(2) = \frac{ b_1 \xi_1 + b_2\xi_2 + (\beta_2 - \beta_1)\sqrt{a_1 b_1 \xi_1}}{b_1 \xi_1 + b_2\xi_2 +\mathcal{P}_{\textrm{tot}}}.
\]
As $\tilde{\kappa}_i  \rightarrow \infty$, we have $a_1, b_1 \rightarrow \infty$. This implies $(\beta_2 - \beta_1)\sqrt{a_1 b_1 \xi_1} > \mathcal{P}_{\textrm{tot}}$. With $\beta_2 > \beta_1$, $f(2) > 1$ and $\mathcal{S}_0 = \{1\}$. %This indicates that the optimal amplifier gains are $g_{1}^2 = \mathcal{P}_{\textrm{tot}}/\xi_1$ and $g_{i} = 0~(i > 1)$. 
Then, $\textrm{P}_{e}^{(\textrm{opt})}({\kappa}_{\infty}) = Q\left( \frac{1}{2 \tilde{\sigma}_v}( \mathcal{P}_{\textrm{tot}} \max\{\theta_i\} )^{1/2} \right)$, where $\theta_i = \gamma_i^2 |h_i|^2/\xi_i$. Furthermore, $\textrm{P}_{e}^{(\textrm{sub})}({\kappa}_{\infty}) = Q\left( \frac{1}{2 \tilde{\sigma}_v}( \mathcal{P}_{\textrm{tot}} \|\boldsymbol{\theta}\|^2/(\boldsymbol{1}^{\texttt{T}}\boldsymbol{\theta}) )^{1/2} \right)$ and $\textrm{P}_{e}^{(\textrm{equ})}({\kappa}_{\infty}) = Q\left( \frac{1}{2 \tilde{\sigma}_v}( \mathcal{P}_{\textrm{tot}} (\boldsymbol{1}^{\texttt{T}}\boldsymbol{\theta})/n )^{1/2} \right)$. 
Since $ \max\{\theta_i\} \cdot (\boldsymbol{1}^{\texttt{T}}\boldsymbol{\theta}) \geq \|\boldsymbol{\theta}\|^2
$ and $n\|\boldsymbol{\theta}\|^2 \geq (\boldsymbol{1}^{\texttt{T}}\boldsymbol{\theta})^2$, we can conclude the proof.
\end{proof}
%\vspace{-0.15in}
%-----------------------------------------
\subsection{Solution for Set $\mathcal{S}_1$}\label{proof:solu_s1}
%-----------------------------------------

Similar to the solution for $\mathcal{S}_0$, we need to show that: given $\tilde{\beta}_1 \leq \tilde{\beta}_2 \leq \cdots \leq \tilde{\beta}_{n}$, $\tilde{f}(i_\mathcal{S}) < 1$ and $\tilde{f}(i_\mathcal{S}+1) \geq 1$, we have
%\begin{enumerate}
%	\item $\tilde{f}(i) < 1$ for $i \leq i_{\mathcal{S}}$;
%	\item $\tilde{f}(i) \geq 1$ for $i_{\mathcal{S}} < i \leq \big \lfloor\frac{\mathcal{P}_{\textrm{tot}}}{\mathcal{P}_{\max}} \big\rfloor$.
%\end{enumerate}
\\$~~~~$\emph{Property F.1}: $\tilde{f}(i) < 1$ for $i \leq i_{\mathcal{S}}$; \\
	$~~~~$\emph{Property F.2}: $\tilde{f}(i) \geq 1$ for $i_{\mathcal{S}} < i \leq \big \lfloor\frac{\mathcal{P}_{\textrm{tot}}}{\mathcal{P}_{\max}} \big\rfloor$.
\begin{proof}
To prove Property F.1, we consider 4 cases which depend on the values of $\tilde{\beta}_i$ and $\beta_i$: 1) $\tilde{\mathcal{S}}_{i} = \tilde{\mathcal{S}}_{i-1} \cup \tilde{\mathcal{S}}'_{i} \setminus \{i \}$, 2) $\tilde{\mathcal{S}}_{i} = \tilde{\mathcal{S}}_{i-1} \cup \tilde{\mathcal{S}}'_{i}  $, 3) $\tilde{\mathcal{S}}_{i} = \tilde{\mathcal{S}}_{i-1}  \setminus \{i \}$, 4) $\tilde{\mathcal{S}}_{i} = \tilde{\mathcal{S}}_{i-1} $, where $\tilde{\mathcal{S}}'_{i} = \{m|\tilde{\beta}_{i-1}< \beta_m < \tilde{\beta}_{i},~i<m\leq n\}$. Now we start with case 1). In case 1), we have $\beta_{i} < \tilde{\beta}_{i-1}$ and $\tilde{\mathcal{S}}'_{i} \neq \emptyset$. Furthermore, we note that 
 \begin{flalign*}
~&\textstyle \tilde{\beta}_{i-1} \sum_{m \in \tilde{\mathcal{S}}_{i-1}} \sqrt{a_m b_m\xi_m} \\ \leq~& \textstyle \tilde{\beta}_{i} \Big(\sum_{m \in \tilde{\mathcal{S}}_{i}} \sqrt{a_m b_m\xi_m} + \sqrt{a_{i}b_i \xi_i }-   \sum_{m \in \tilde{\mathcal{S}}'_{i}} \sqrt{a_m b_m\xi_m}\, \Big)\\
\leq~& \textstyle \tilde{\beta}_{i} \sum_{m \in \tilde{\mathcal{S}}_{i}} \sqrt{a_m b_m\xi_m} + (\mathcal{P}_{\max} + b_i \xi_i) - \sum_{m \in \tilde{\mathcal{S}}'_{i}} b_m\xi_m.
\end{flalign*} 
The last inequality is valid since when $m \in \tilde{\mathcal{S}}'_{i}$, $\beta_m <\tilde{\beta}_{i}$, we have \[\tilde{\beta}_{i} \sum_{m \in \tilde{\mathcal{S}}'_{i}} \sqrt{a_m b_m\xi_m} \geq \sum_{m \in \tilde{\mathcal{S}}'_{i}} \beta_m\sqrt{a_m b_m\xi_m}  = \sum_{m \in \tilde{\mathcal{S}}'_{i}} b_m\xi_m. \] 
After some manipulations, when $\tilde{f}(i) < 1,~\forall i < i_{\mathcal{S}}$,
% \begin{flalign*}
%\tilde{f}(i-1) ~=~& \frac{\tilde{\beta}_{i-1}\sum_{m \in \tilde{\mathcal{S}}_{i-1}} \sqrt{a_m b_m\xi_m}}{\mathcal{P}_{\textrm{tot}} - ({i-1}) \mathcal{P}_{\max} + \sum_{m \in \tilde{\mathcal{S}}_{i-1}} b_m\xi_m} \\
%~\leq~& \frac{\tilde{\beta}_{i}\sum_{m \in \tilde{\mathcal{S}}_i} \sqrt{a_m b_m\xi_m} +c_1}{(\mathcal{P}_{\textrm{tot}} - i \mathcal{P}_{\max} + \sum_{m \in \tilde{\mathcal{S}}_i} b_m\xi_m) + c_1} \\
%~<~& 1,
%\end{flalign*} 
 \begin{flalign*}
\tilde{f}(i-1)~\leq~   \frac{\tilde{\beta}_{i}\sum_{m \in \tilde{\mathcal{S}}_i} \sqrt{a_m b_m\xi_m} +c_1}{(\mathcal{P}_{\textrm{tot}} - i \mathcal{P}_{\max} + \sum_{m \in \tilde{\mathcal{S}}_i} b_m\xi_m) + c_1} ~<~ 1,
\end{flalign*} 
where $c_1 = (\mathcal{P}_{\max} + b_i \xi_i) - \sum_{m \in \tilde{\mathcal{S}}'_{i}} b_m\xi_m$. The last inequality is valid because when $x/y < 1$, we have $(x+c_1)/(y+c_1) < 1$, where $x,y > 0$ and $c_1 > -x$. Similarly, we see that for other three cases, we also have $\tilde{f}(i-1) < 1$.

Now let us prove Property F.2.  Similar to Property F.1, we have
 
 \begin{flalign*}
\tilde{f}(i+1) ~\geq~  \frac{\tilde{\beta}_{i}\sum_{m \in \tilde{\mathcal{S}}_i} \sqrt{a_m b_m\xi_m} -c_2}{(\mathcal{P}_{\textrm{tot}} - i \mathcal{P}_{\max} + \sum_{m \in \tilde{\mathcal{S}}_i} b_m\xi_m) - c_2} 
~\geq~ 1.
\end{flalign*} 
where $c_2 = (\mathcal{P}_{\max} + b_{i+1} \xi_{i+1}) - \sum_{m \in \tilde{\mathcal{S}}'_{i+1}} b_m\xi_m$. The last inequality is valid because when $x/y \geq 1$, we have $(x-c_2)/(y-c_2) \geq 1$, where $x,y > 0$ and $c_2 < y$. Similarly, we see that for other three cases, we have $\tilde{f}(i+1) \geq 1$. This completes the proof.
\end{proof}
 %\vspace{-0.15in}

 \bibliographystyle{unsrt}
\bibliography{IEEEabrv,Spectrum_Sensing_All}

\end{document}